\documentclass[12pt, draftclsnofoot, onecolumn]{IEEEtran}

\IEEEoverridecommandlockouts 

\usepackage{amssymb, cite}
\usepackage{bm,bbm}
\usepackage{mathtools}
\usepackage[size=small]{caption}
\usepackage{multirow,color,amsfonts}
\usepackage{tabulary}
\usepackage{subfigure}
\usepackage{graphicx}
\usepackage{setspace}
\usepackage{enumerate}
\usepackage[]{algorithm2e}
\usepackage{comment}

\usepackage[utf8]{inputenc} 
\usepackage[T1]{fontenc}
\usepackage{ifthen}
\usepackage{cite}

\usepackage[hyphens]{url}
\usepackage[hidelinks]{hyperref}
\hypersetup{breaklinks=true}
\usepackage{amsthm}
\usepackage{amsmath}

\newtheorem{theo}{Theorem}
\newtheorem{defi}{Definition}

\newtheorem{prop}{Proposition}

\newtheorem{cor}{Corollary}

\newtheorem{lem}{Lemma}

\newtheorem{exm}{Example}
\usepackage[dvipsnames]{xcolor}

\newcommand{\expan}{E_{\theta}}
\newcommand{\subg}{S}

\begin{document}
\title{Distributed Compression for Computation and Bounds on the Optimal Rate}

\author{Mohammad~Reza~Deylam~Salehi and Derya~Malak\thanks{M.~R.~Deylam~Salehi and D. Malak are with the Communication Systems Department, EURECOM, Biot Sophia Antipolis, FRANCE (Emails: \{deylam, malak\}@eurecom.fr). An early version of this work appeared in Proc., Allerton 2023~\cite{10313467}. This work is partially supported by a Huawei France-funded Chair towards Future Wireless Networks. 

Co-funded by the European Union (ERC, SENSIBILITÉ, 101077361). Views and opinions expressed are however those of the authors only and do not necessarily reflect those of the European Union or the European Research Council. Neither the European Union nor the granting authority can be held responsible for them.} 
}
\maketitle
 
\begin{abstract} 
We address the problem of distributed computation of arbitrary functions of two correlated sources $X_1$ and $X_2$, residing in two distributed source nodes, respectively. We exploit the structure of a computation task by coding source characteristic graphs (and multiple instances using the $n$-fold OR product of this graph with itself). For regular graphs and general graphs,  we establish bounds on the optimal rate---characterized by the chromatic entropy for the $n$-fold graph products---that allows a receiver for asymptotically lossless computation of arbitrary functions over finite fields. For the special class of cycle graphs (i.e., $2$-regular graphs), we establish an exact characterization of chromatic numbers and derive bounds on the required rates. Next, focusing on the more general class of $d$-regular graphs, we establish connections between $d$-regular graphs and expansion rates for $n$-fold graph powers using graph spectra. Finally, for general graphs, we leverage the Gershgorin Circle Theorem (GCT) to provide a characterization of the spectra, which allows us to build new bounds on the optimal rate. Our codes leverage the spectra of the computation and provide a graph expansion-based characterization to efficiently/succinctly capture the computation structure, providing new insights into the problem of distributed computation of arbitrary functions.
\end{abstract}


\section{Introduction}
\label{sec:intro}

Data compression is the process of using fewer bits than the original size of the source, which is given by Shannon's entropy of a source~\cite{shannon1948mathematical} in the case of a point-to-point communication model. In the setting of distributed sources,  where the goal is to recover them jointly at a receiver, the Slepian-Wolf theorem gives the fundamental limits of rate for compression~\cite{slepian1973noiseless}. In the case when the receiver wants to compute a deterministic function of distributed sources, a further reduction in compression is possible via accounting for the structure of the function as well as the structure of the joint source distribution~\cite{AO96}. This approach is known as distributed functional compression, in which a function represents an abstraction of a particular task, the sources separately compress their data and send \emph{color encodings} of their data to a common receiver, and the receiver, from the obtained colors, recovers the desired function of the sources. This approach differs from the conventional approach~\cite{slepian1973noiseless}, where the receiver jointly decodes source sequences.

\subsection{Motivation and Literature Review}
Let us begin with an example. Consider a college student database with information including the rental records, demographics and health of individuals. The Ministry of Science wants to offer housing aid to a particular group of students by requiring information solely on the rental contracts and payslips of the students, and without disclosing their personal data, due to privacy and redundancy constraints. 
In data compression, storing and sending a stream of data that is not always required is not practical from the perspectives of storage and transmission. 
The scenario of student housing is an example of realizing functional compression, which avoids compressing and transmitting large volumes of all available data, and is instead tailored to the specifics of the function, i.e., a student's eligibility for getting housing aid or not.

In Shannon's breakthrough work in \cite{shannon1948mathematical}, the function to be recovered at the receiver is the identity function of the source variable, i.e., the source itself. Generalizing the noiseless coding of a discrete information source, given in~\cite{shannon1948mathematical}, to distributed compression and joint decoding of two jointly distributed and finite alphabet $X_1$ and $X_2$, Slepian-Wolf theorem gives a theoretical lower bound for the lossless coding rate of distributed coding of such sources~\cite{slepian1973noiseless},  
where the two data sequences of memoryless correlated sources with finite alphabets ${\bf X}_1^n$, and ${\bf X}_2^{n}$ are obtained by $n$ repeated independent drawings from a discrete bivariate distribution. Practical implementation schemes for Slepian-Wolf compression have been proposed, including \cite{pawlak2003postfiltering,ho2006random,ahlswede2000network}. In functional compression of distributed sources $X_1$, and $X_2$, however, the goal is to compress the sources separately while ensuring that a deterministic function $f(X_1, X_2)$ of these sources can be calculated by a user.
Prior attempts on functional compression can be categorized into works focusing on lossless and zero error compression of functions,~\cite{shannon1956zero,shannon1948mathematical, AO96,doshi2007distributed,doshi2010functional,slepian1973noiseless, OR01, feizi2014network, malak2022fractional,korner1979encode,ahlswede1975source,coleman2006low,han1987dichotomy,pradhan2003distributed, malak2024multi, malak-func, tanha2024influence,khalesi2024perfect, witsenhausen1976zero}, and those for which the compression schemes tolerate distortion for lossy reconstruction~\cite{wyner1976rate,feng2004functional,yamamoto1982wyner,berger1989multiterminal,barros2003rate,wagner2008rate, rebollo2009t}.

Several special cases of distributed compression have been studied. In~\cite{ahlswede1975source}, Ahlswede and K{\"o}rner have determined the rate region for a distributed compression setting where separate encoders encode the $n$-th realizations of correlated sources $X_1$ and $X_2$ observed by sources one and two respectively, and a receiver aims to only recover ${\bf X}_2^n$. 
K{\"o}rner and Marton, in~\cite{korner1979encode}, have dealt with the problem of distributed encoding of two binary sources $X_1$ and $X_2$ to compute their modulo-two sum, i.e., $f(X_1, X_2)=(X_1+X_2)\mod 2$, at the receiver. 
In~\cite{yamamoto1982wyner}, Yamamoto has extended the Wyner-Ziv model~\cite{wyner1976rate}, which has addressed lossy source coding with side information at the decoder, to a setting where the decoder estimates a function $f(X_1, X_2)$ of the source $X_1$ given side information $X_2$. In~\cite{han1987dichotomy}, Han and Kobayashi have established an achievable distributed functional reconstruction scheme, which depends on the structure of $f(X_1, X_2)$, and the joint distribution of $(X_1, X_2)$.
Building on~\cite{yamamoto1982wyner}, optimal coding schemes and achievable rate regions have been derived for lossless and lossy compression of source $X_1$ for distributed computation of $f(X_1, X_2)$ given side information $X_2$~\cite{OR01,shirani2021new,doshi2010functional, yuan2022lossy}, for distributed compression of sources $X_1$ and $X_2$ for the computation of $f(X_1, X_2)$~\cite{feizi2014network, SefTch2011}. 
More specifically, in~\cite{OR01}, Orlitsky and Roche have provided a single letter characterization for general functions of two variables using the notion of~\emph{source characteristic graphs} (confusion graphs) introduced by K{\"o}rner~\cite{korner1973coding}, where the vertices are the possible realizations of a source and the edges capture the function structure. 

\subsection{Overview and Contributions}
\label{contribution}

In this work, we design a coding framework for the problem of distributed functional compression with two distributed sources having access to $X_1$ and $X_2$, respectively, each with a finite alphabet, and a receiver that wants to reconstruct the function $f(X_1, X_2)$ in an asymptotically lossless manner. To capture the structure of the function $f$ in compression, we exploit the notion of \emph{source characteristic graphs}. Given source variable $X_1$, the limits of color reuse are determined by chromatic numbers $\chi(G)$~\cite{witsenhausen1976zero}\footnote{In~\cite{witsenhausen1976zero}, Witsenhausen has considered zero-error compression using characteristic graphs with side information.}, and the fundamental limits of compression rate are characterized by K{\"o}rner's graph entropy $G_{X_1}$~\cite{korner1973coding}.
To address different computation scenarios, we examine several characteristic graph topologies, including cyclic graphs, denoted by $C_{i}$ where $i$ is the number of vertices, their generalizations to $d$-regular graphs denoted by $G_{d, V}$, and general characteristic graphs denoted by $G$, as motivated next.

 
Cycle graphs (or cyclic graphs) appear in many practical scenarios, such as periodic functions, and mod functions, which are widely used in cryptography, computer algebra and science, and musical arts \cite{pettofrezzo1970elements}. In cryptography, Caesar Ciphers, Rivest-Shamir-Adleman (RSA) algorithm~\cite{paar2009understanding}, Diffie-Hellman \cite{diffie1976multiuser}, as well as Advanced Encryption Standard (AES) \cite{selent2010advanced}, International Data Encryption Algorithm (IDEA), are widely used for secure data transmission~\cite{basu2011international}. The calculation of checksums within serial numbers is another application of interest \cite{knill1981applications}. For example, ISBNs (International Standard Book Numbers) use mod 11 arithmetic for 10-digit ISBNs or mod 10 for 13-digit ISBNs to detect errors. 
In addition, International Bank Account Numbers (IBANs) use mod 97 arithmetic to identify mistakes in bank account numbers entered by users. Cyclic characteristic graphs allow for a more efficient reuse of colors compared to acyclic graphs with a higher average degree. 
Furthermore, cycles and their products --- built to capture a source sequence ${\bf X}_1^n$ --- have good connectivity properties, enabling an exact characterization of their chromatic numbers (and bounding their graph entropies).

$d$-regular graphs have broad applications ranging from representing network topologies to modeling social networks, coding theory for constructing error-correcting codes~\cite{friedman2005generalized}, random walks and Markov chains in analyzing state transitions~\cite{berestycki2018random}, spectral graph theory providing insights into graphs' characteristics~\cite{kahale1991better, kahale1992second}, and fault-tolerant systems \cite{friedman2005generalized, berestycki2018random, felber2013survey, ganesh2005effect,tentes2009expander, SZABO200797}. In general, $d$-regular graphs help model frameworks for structured data, such as graph neural networks~\cite{you2021identity}, making them one of the main topics of investigation in this paper. 

The main contributions of this paper can be summarized as follows:

\begin{itemize}
\item {\emph{Cyclic characteristic graphs:}} We derive exact expressions for the degree of a vertex ${x}^n\in [V^n]$ in the $n$-fold OR product $C_i^n$ (as detailed by Alon and Orlitsky in~\cite{AO96}) of cycles $C_i$ (where $i=V$, see Proposition~\ref{prop-even-cycle-chorom}), denoted by $deg(x)$, and for the chromatic number of even cycles $C_{2k}$, denoted by $\chi({C_{2k}})$ where $k\in \mathbb{Z}^+$. Then, we devise a polynomial-time\footnote{Finding a minimum entropy coloring in general graphs is an NP-hard problem~\cite{cardinal2008tight}.} achievable coloring scheme for odd cycles $C_{2k+1}$, leveraging the structure of $C_i$, and its OR products (see Proposition~\ref{prop-chromatic_number_cycles-odd}). Given $C_i$, we investigate the largest eigenvalue of its adjacency matrix, and using that, we present bounds on the chromatic number~(see Proposition~\ref{prop-eig-bound-chro-cycle}). We also provide bounds on K{\"o}rner's graph entropy of $C_i$ (see Proposition~\ref{prop-upper_entrop-cycle}).

\item {\emph{$d$-regular characteristic graphs:}} We characterize the exact degree of a vertex and the chromatic number of $d$-regular graphs, denoted by $G_{d, V}$, and their $n$-fold OR products $G_{d, V}^n$ (see Propositions~\ref{prop-deg-regular} and~\ref{prop-chrom-d-regular}). Additionally, given a $d$-regular graph, the concept of graph expansion helps determine how the corresponding OR products are related. Capturing the structure of the OR products graphs, we then present a lower bound on the expansion rate of $G_{d, V}^n$ (see Proposition~\ref{prop-expansion-G_dv}).

\item {\emph{General characteristic graphs:}} Given a general graph, $G(\mathcal{V},\mathcal{E})$, we calculate the degree of each vertex for its $n$-fold OR product (see Corollary~\ref{cor-degree-or-power-general-graphs}). We present upper and lower bounds on the expansion rate (see Corollary~\ref{cor-expansion_upper_bound_and _lower_bound}). We then investigate the entropy of general characteristic graphs (see Proposition~\ref{prop_entropy_odd_cycles_MIS}). We derive bounds on the largest eigenvalue (see Corollary~\ref{Cor-bound-chor-using-deg-max-min}), the chromatic number~(see Corollary~\ref{cor-general_graph-approximation}), using the adjacency matrix of the $n$-fold OR product graph and the famous \emph{Gershgorin Circle Theorem} (GCT), which is a theorem that identifies the range of the eigenvalues for a given square matrix~\cite{tretter2008spectral}. We use GCT to bound eigenvalues of the adjacency matrix of a given graph $n$-fold OR product, via exploiting the structure of OR products (see Theorem~\ref{theo-eig-any-func}, and Corollary~\ref{cor:iter-gresh-eigen}). 
\end{itemize}

By leveraging chromatic numbers and graph entropy bounds, our results, as outlined above, illustrate the connection between graph structures and entropy-based communication cost and provide insights for functional compression in distributed settings.

\subsection{Organization}
\label{org}
The rest of this manuscript is organized as follows. In Section~\ref{sec:model}, we review the literature and provide a technical preliminary on graphs, their valid colorings, and the $n$-fold OR products of characteristic graphs. In Section~\ref{sec:results}, we present the main results of the paper, including the achievable coloring schemes, the bounds on the degrees, eigenvalues, and the chromatic numbers for cycles, $d$-regular graphs, and general graphs. We further derive upper and lower bounds on the graph entropy and expansion rate for the $n$-fold OR product of characteristic graphs. 
In Section~\ref{sec:conclusion}, we summarize our key results and outline potential exploration directions. Proofs of the main results are presented in the Appendix.

\subsection{Notation}
\label{notation}
Letter $X$ denotes a discrete random variable with distribution $p(x)$ over the finite alphabet $\mathcal{X} $, where $x$ is a realization of $X$, and ${\bf X}^n=X_1, X_2 \cdots, X_n$ is an independent and identically distributed (i.i.d.) sequence where each element is distributed according to $p(x)$. We denote matrices and vectors by boldface letters, e.g., ${\bf A }$. 
The joint distribution of the source variables $X_1$ and $X_2$ is denoted by
$p(x_1,x_2)$. We denote by $G_{X_1}$ and $G_{X_2}$ the characteristic graphs that sources $X_1$ and $X_2$ build for computing a given function $f(X_1, X_2)$, respectively.
We use the boldface notation ${\bf x}_1^n=x_{11},\dots, x_{1n}$, to represent the length $n$-th realization of $X_1$, and similarly for $X_2$. We let $[n]=\{1,\dots,n\}$ for $n\in\mathbb{Z}^+$, and $[a,b]=\{a, a+1,\cdots, b\}$  for $a,b\in\mathbb{Z}^+$.

Given a graph $G(\mathcal{V}, \mathcal{E})$, we denote by $\chi({G})$ and $\chi_f(G)$ its chromatic and fractional chromatic numbers, respectively. We denote by $deg({x_k})$ the degree of the vertex ${x_k}\in [V]$, 
$d_{\max} =\max_{k \in[V]} deg({x}_k)$ is the maximum vertex degree of $G$, and $d_{\rm avg}(x_k)$ denotes the average degree over $x_k \in [V]$ in $G^1$. We denote by $C_i=G(\mathcal{V},\,\mathcal{E})$ a cycle graph with $i=V$ vertices, by $C_{i}^{n}=G(\mathcal{V}^n, \mathcal{E}^n)$ its $n$-fold OR product, and by $\mathcal{C}^j_{i}(l)$ the set of distinct colors in sub-graph $l\in [V]$ of the $j$-fold OR product of $C_i$. 
We denote the coloring distribution of $G$ by $\mathcal{C}_{G}$, and the set of distinct colors of $G$ by $\mathcal{C}(G)$. We denote a $d$-regular graph on $V$ vertices by $G_{d, V}$. We denote a complete graph with $i=V$ vertices and its $n$-fold OR product by $K_i$ and $K_i^n$, respectively. We denote $\mathcal{C}_{G}$ as the PMF of a valid coloring of a graph $G$.

Given a graph $G(\mathcal{V},\mathcal{E})$, we denote by ${\bf J}_V$ and ${\bf I}_V$ an all-one and identity matrices of size $V\times V$ each, by ${\bf A}_f$ the adjacency matrix, where ${\bf A}_f=(a_{xx'})_{1\leq { x},\,{x}'\leq V}$ is a symmetric $(0,1)$-matrix with zeros on its diagonal, i.e., $a_{xx}=0$, and  $a_{xx'}=1$ indicates that two distinct vertices ${x},\,{x}'\in [V]$ are adjacent, and $a_{xx'}=0$ when there is no edge between them. 
The trace of ${\bf A}_f$ is denoted by $trace({\bf A}_f)$. The largest and the smallest eigenvalues of ${\bf A}_f$ are denoted by $\lambda_1({\bf A}_f)$, and $\lambda_V({\bf A}_f)$, respectively, $\Xi ({\bf A}_f)$ is the set of all eigenvalues of ${\bf A}_f$, and $\vartheta(G^{j})$ is the set of distinct eigenvalues of the adjacency matrix of $G^j$, i.e., ${\bf A}_{f}^{j}$. 
Exploiting GCT to characterize the eigenvalues of ${\bf A}_f$, the $k$-th interval that contains an eigenvalue is denoted by $\delta_k$, $k \in [V]$. LHS and RHS represent the left and right-hand sides of an equation, respectively.


\section{Technical Preliminary}
\label{sec:model}

This section introduces the fundamental concepts related to graph theory, such as degrees, independent sets, paths, cycles, $d$-regular graphs, and the notion of graph expansion~\cite{havel1955remark,beigel1999finding,chen1997graph, korner1973coding,nagle1966ordering}. Furthermore, it discusses the concepts of characteristic graphs, the $n$-fold OR products of characteristic graphs, and traditional and fractional coloring of graphs~\cite{nagle1966ordering,nilli1991second,alon2002graph}.

\subsection{Source Characteristic Graphs and Their OR Products}
\label{characteristic_graphs}

A graph is represented by $G(\mathcal{V}, \mathcal{E})$, where $\mathcal{V}=[V]$ denotes the set of its vertices, with cardinality $|\mathcal{V}|=V$, and $\mathcal{E}$ is the set of its edges, with cardinality $|\mathcal{E}|=E$.

\begin{defi}[Degree of a vertex \cite{havel1955remark}]
\label{def-Degree-of-a-vertex}
Given $G(\mathcal{V}, \mathcal{E})$, the degree of a vertex ${x}_k\in[V]$ for $k\in[V]$, represented by $deg({x}_k)$, is the number of edges it is connected to, i.e., the number of neighbors of ${x_k}\in [V]$. The average degree across nodes in $G$ is denoted by ${d}_{\rm avg}=\frac{\sum_{x_k \in [V]} deg (x_k)}{V}$.  
\end{defi}

We next introduce the concept of an independent set, which plays a critical role in determining a valid coloring of characteristic graphs that we detail in Section~\ref{coloring}.

\begin{defi}[Independent set, and maximal  independent set \cite{beigel1999finding}]
\label{Independent set, and maximum independent set}
An independent set, $\text{IS}_G$, in $G(\mathcal{V}, \mathcal{E})$ is a subset of vertices of $\mathcal{V}$, such that no two are adjacent. A maximal independent set, $\text{MIS}_G$, is an independent set in $G$ that is not a subset of any other independent set of $G$.  
\end{defi}


In distributed functional compression with $M$ source nodes, each holding $X_k\in\mathcal{X}_k$, $k\in[M]$, a receiver aims to reconstruct $f(X_1,X_2,\dots,X_M)$. To aid in distinguishing function outcomes, each source $k$ builds a characteristic graph $G_{X_k}$ with vertex set $\mathcal{X}_k$ and edges determined by the function and the joint source PMF. Next, we define characteristic graphs for a bivariate setting.

\begin{defi}[Source characteristic graphs \cite{korner1973coding}]
\label{characteristic graphs}
Let $X_1$ and $X_2$ be two distributed source variables with a joint distribution $p(x_1,x_2)$. 
Source one builds a characteristic graph $G_{X_1}=G(\mathcal{V}, \mathcal{E})$ for distinguishing the outcomes of a function $f(X_1,\, X_2)$, where $\mathcal{V}=\mathcal{X}_1$, and an edge $(x_1^1, x_1^2)\in  \mathcal{E}$ if and only if there exists a $x_2^1 \in \mathcal{X}_2$ such that $p(x^1_1, x^1_2)\,\cdot\, p(x^2_1, x^1_2) > 0 $ and $f(x^1_1, x^1_2)\neq f(x^2_1, x^1_2)$, i.e., these two vertices of $G_{X_1}$ should be distinguished. 
\end{defi}


\begin{defi}[Entropy of a characteristic graph~\cite{korner1973coding}]
Given a source random variable $X_1$ with  characteristic graph $G_{X_1}=G(\mathcal{V}, \mathcal{E})$, the entropy of $G_{X_{1}}$ is defined as
\begin{align}
\label{eq-graph_entropy-korner}
H_{G_{X_1}}(X_1)= \min\limits_{X_1\in U_1\in \text{MIS}_{G_{X_1}}} I(X_1; U_1)\ ,
\end{align}
where $\text{MIS}_{G_{X_1}}$ represents the set of all MISs of $G_{X_1}$ \cite{AO96}. The notation $X_1 \in U_1 \in \text{MIS}_{G_{X_1}}$ indicates that the minimization is performed over all distributions $p(u_1, x_1)$ such that $p(u_1, x_1) > 0$ implies $x_1 \in u_1$, where $U_1$ is an MIS of $G_{X_1}$.
\end{defi}

From (\ref{eq-graph_entropy-korner}), it follows that $H_{G_{X_1}}(X_1)\leq H(X_1)$, yielding savings over~\cite{shannon1948mathematical}. Next, we introduce a path, which refers to a sequence of edges connecting a subset of vertices within a graph.

\begin{defi}[Path and Hamiltonian path~\cite{bermond1979hamiltonian}]
\label{def-path}
Given an undirected graph $G(\mathcal{V}, \mathcal{E})$, a path is a sequence of vertices starting and ending with distinct vertices, where each pair of consecutive vertices is connected by an edge, and no vertex is repeated. A path that includes every vertex of a graph exactly once is called a Hamiltonian path.
\end{defi}

We next define the class of $d$-regular graphs that embeds the special case of cycles. 

\begin{defi}[$d$-regular graphs~\cite{chen1997graph,nagle1966ordering}]
\label{d-regular graph, and Cycle graphs}
A $d$-regular graph $G_{d, V}(\mathcal{V}, \mathcal{E})$ is a graph where each vertex has the same degree $d$, i.e., $d=deg(x_k)$ for all $x_k\in [V]$. 
A $d$-regular graph $G_{d, V}$, where $d, V \in \mathbb{Z}^+$, satisfies $V \geq d+1$. Furthermore, if $d$ is odd, the total number of vertices $V$ must be even \cite{axenovich2014lecture}. Cyclic graphs are $2$-regular graphs with a Hamiltonian path.  
\end{defi}

We next define the expansion rate of graphs.

\begin{defi}[Expansion rate~\cite{bang2008digraphs}]
\label{defini-expansion}

An undirected expander graph $G$ is a graph having relatively few edges in comparison to its number of vertices while maintaining {\emph{strong connectivity properties}},  which ensures that each vertex is reachable by paths from at least $2$ directions~\cite{bang2008digraphs}. The expansion of $G$ with respect to a subset of its vertices $Y \subseteq [V]$ is determined as follows:
\begin{align}
\label{eq-graph_expansion_formula}
\expan(G)=\frac{|N_G(Y)|}{|Y|} \ ,
\end{align}
where $N_G(Y)=\{u \in [V], u \notin Y: \exists v \in Y, (v,u) \in\mathcal{E}\}$ denotes the set of neighbors of $Y$.
\end{defi}
We next illustrate the concept of characteristic graphs with an example.

\begin{exm}
\label{characteristic graph ex} 
Consider the problem of distributed functional compression of $f(X_1, X_2)=(X_1+X_2)\mod 2$, with two source variables $X_1$ and $X_2$ and one receiver, where $X_1$ is uniform over the alphabet $\mathcal{X}_1=\{0,1,2,3\}$, and $X_2$ is uniform over $\mathcal{X}_2=\{0,1\}$. 

For even-valued, i.e., $X_1\in \{0,2\}$, the output is $f(X_1,X_2)=X_2$, and for odd values, i.e., when $X_1\in \{ 1,3\}$, we have $f(X_1,X_2)=(X_2+1)\mod 2$. In the characteristic graph built for $X_1$, namely $G_{X_1}$, we do not need to distinguish the elements of $\{0,2\}$ from each other, and similarly for the elements of $\{1,3\}$. However, these two sets must be distinguished, which is possible via using two distinct colors $B$ and $O$. To that end,  we assign the elements of $\{0,2\}$ and $\{1,3\}$ colors $B$ and $O$, respectively. Similarly, the outcome $X_2=1$ is assigned $R$, and $X_2=0$ is assigned $Y$. We illustrate the coloring of $G_{X_{1}}$ and $G_{X_2}$ in Figure~\ref{fig:funcom}. Transmission of an assigned color pair from distributed sources to a common receiver according to the described rule is sufficient for the receiver to determine the corresponding outcome of $f$ via a look-up table~\cite{feizi2014network}. The scheme satisfies the necessary condition since both $G_{X_1}$ and $G_{X_2}$ require at least 2 colors (1 bit per source). Using fewer colors by assigning the same color to adjacent vertices violates the condition $f(x^1_1, x^1_2) \neq f(x^2_1, x^1_2)$ when $p(x^1_1, x^1_2)\cdot p(x^2_1, x^1_2)> 0$.
\end{exm}
 
\begin{figure}[h!]
\centerline{\includegraphics[scale=0.25]{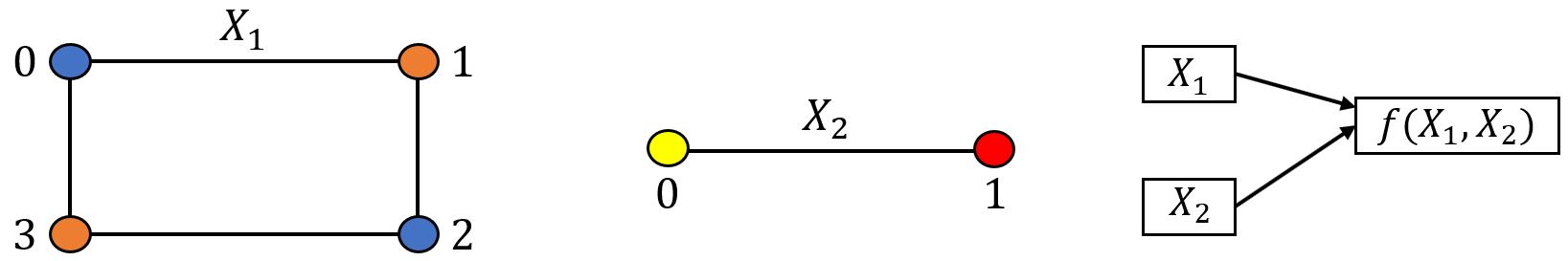}}
\caption{Distributed functional compression with two sources and a receiver, where $G_{X_1}$ is cyclic.}
\label{fig:funcom}
\end{figure}

To determine the fundamental limits of lossless compression of a source sequence~\cite{AO96}, we exploit the notion of \emph{$n$-fold OR products} of graphs\footnote{Here, we exclusively focus on the OR product graphs for realizing lossless compression of multi-letter schemes~\cite{korner1998zero}. If the goal is instead to perform zero-error source compression, one needs to employ the $n$-fold AND products of graphs~\cite{koulgi2003zero, tuncel2007kraft}.}, as introduced next.

Given a source variable $X$, we next detail the construction for the $n$-fold OR products graphs, by generalizing the rule in Definition~\ref{characteristic graphs} for building $G_X$ to multiple source instances ${\bf X}^n$.

\begin{defi}[$n$-fold OR product graph \cite{nagle1966ordering,nilli1991second,alon2002graph}]
\label{OR power}
For $n>1$, the $n$-fold OR product of $G_{X}=G(\mathcal{V}, \mathcal{E})$ is represented as $G_{X}^n=G(\mathcal{V}^{n},\mathcal{E}^{n})$, where $\mathcal{V}^n=\mathcal{X}^n$, and given two distinct vertices ${\bf x}_1^n=x_{11}^1,\dots, x_{1n}^1\in[V^n]$, and ${\bf x}_2^n=x_{21}^1,\dots, x_{2n}^1 \in[V^n]$, it holds that $({\bf x}_1^n,{\bf x}_2^n)\in \mathcal{E}^n$, when $\exists$ at least one $k\in [n]$ such that $({x}_{1k}^{1}, {x}_{1k}^{2})\in \mathcal{E}$.  
\end{defi}

$n$-fold OR products have been extensively used in asymptotically lossless compression of distributed sources for computing functions, see e.g., \cite{feizi2014network, malak2022fractional, OR01}. Building on \cite{korner1973coding}, the authors in \cite{OR01} have demonstrated that given two distributed source variables $X_1$ and $X_2$, the lowest sum rate needed for distributed computing of $f(X_1, X_2)$ could be achieved by encoding the $n$-fold OR product graphs $G_{{X}_1}^n$ and $G_{{X}_2}^n$, in the limit as $n$ goes to infinity.

We next describe how to determine the $n$-fold OR product $G_{X}^n$ of a characteristic graph iteratively, from the $(n-1)$-fold OR product graph $G_{X}^{n-1}$.

\begin{defi}[Sub-graphs of $G_{{X}}^n$]
\label{def-sub-graph}
Given an $n$-fold OR product graph $G_{X}^n$, we denote the collection of graphs $\{G_{X}^{n}(l)\}_{l\in [V]}$ in $G_{X}^n$ as the sub-graphs of $G_{X}^n$, where each of $G_{X}^n(l)$ represents $(n-1)$-fold OR product of $G_X$ with itself.
\end{defi}

We next define \emph{'full connection'} that refers to links between sub-graphs in an OR product.

\begin{defi}[Full connection between two graphs]
\label{def-full-conec}
Given $G_1(\mathcal{V}_1,\mathcal{E}_1)$ and $G_2(\mathcal{V}_2,\mathcal{E}_2)$, if for each ${x}_{1k} \in [V_1]$, and each ${x}_{2t} \in [V_2]$, there exists an edge $({x}_{1k}, {x}_{2t})$, i.e., the bipartite graph formed between $[V_1]$ and $[{V}_2]$ is complete, we describe the two graphs as having a full connection.  
\end{defi}

With these principles established, we next detail their application to functional compression.

\subsection{Coloring of Characteristic Graphs}
\label{coloring}

A valid (proper) vertex coloring of $G_{X_1}$ assigns a color to each vertex such that adjacent vertices receive distinct colors, indicating which source realizations need different codes (colors). Non-adjacent vertices may share the same color, which is known as~\emph{traditional graph coloring}.
A valid coloring that achieves the minimum entropy among all valid colorings gives a lower bound to the compression rate for the lossless reconstruction of the desired function. The minimum number of colors required to achieve a valid coloring of $G_{X_1}$ is called the chromatic number\footnote{In general, the problem of determining $\chi({G_{X_1}})$ is NP-complete \cite{cardinal2008tight}.}, $\chi (G_{X_1})$. \emph{Fractional graph coloring} generalizes the concept of traditional coloring by assigning a fixed number of distinct colors from a set of available colors to each ${x}_k\in[V]$ such that adjacent vertices have non-overlapping sets of colors~\cite{scheinerman2011fractional,malak2022fractional}. Moreover, given the connection between the coloring distribution and the minimum entropy of graphs~\cite{korner1973coding}, the fractional chromatic number of $G$, which is a lower bound on the chromatic number of $G$, is important to investigate. 

We next detail how to obtain a valid $b$-fold coloring out of $a$ available colors.

\begin{defi}[The fractional chromatic number]
\label{def-fractional-coloring}    
A valid $b$-fold coloring assigns sets of distinct colors with cardinality $b$ to each vertex such that adjacent vertices receive disjoint sets of $b$ colors. A valid $a:b$ coloring is a valid $b$-fold coloring that uses a total of $a$ available colors. 
The notation $\chi_b({G})$ represents the $b$-fold chromatic number of $G$, with the smallest $a$ number of colors, such that an $a:b$ coloring exists. The fractional chromatic number of $G$ is given as
 \begin{align}
 \label{eq-proof-theorem-converse-fractional-number}
\chi_f({G})=\lim_{b\to \infty} \,\frac{\chi_{b}({G})}{b}=\inf_{b} \,\frac{\chi_{b}({G})}{b} \ ,
 \end{align}
where $\chi_{b}({G}) \in \mathbb{Z}^+$, and the second equality follows from the subadditivity of ${\chi_{b}({G})}$~\cite{scheinerman2011fractional}. 
\end{defi}

Given a graph $G$, the fractional chromatic number of its $n$-fold product is computed as~\cite{malak2022fractional}:
\begin{align}
\label{eq-n-fold-fractional}
\chi_f({G^n}) = (\chi_f({G}))^n \ .
\end{align}
Traditional coloring is a special case of the valid $a:b$ coloring of $G^n$ where $b=1$ and $a=\chi({G^{n}})$.

To investigate a general characteristic graph $G$ and its chromatic number and entropy, we aim to leverage the eigenvalue relationships of its adjacency matrix. Therefore, we define the GCT for square and block matrix representations to address coloring in general characteristic graphs.
\subsection{Gershgorin Circle Theorem (GCT)}
\label{prelem}
We next introduce the GCT for the eigenvalues of square matrices and later establish their connection to the chromatic number in Section~\ref{sec:Eig-Adj-Mat-Cycl-Graph-}.

\begin{defi}[GCT for square matrices\cite{tretter2008spectral}]
 \label{rem-gersh-circle}
Given a square matrix ${\bf A} \in \mathbb{C}^{V \times V}$ with elements $a_{kt}$, where $k,t\in [V]$, we define $D_k$ as a circle that contains the eigenvalue $\lambda_k({\bf A})$ as follows:
\begin{align}
 \label{eq-def-gersh-circle}
 D_k= \{ {\lambda_k({\bf A})} \in \mathbb{C}\: : |{\lambda_k({\bf A})}-a_{kk}|\leq \sum_{t\neq k} |a_{kt}| \} \ ,
\end{align}
where $\sum_{t\neq k} |a_{kt}|$ is the sum of the absolute values of the non-diagonal entries in the $k$-th row of ${\bf A}$, and the set of all eigenvalues $\Xi({\bf A})=\{\lambda_k({\bf A}): k\in [V]\}$ satisfies $\Xi({\bf A}) \in D_1\cup D_2 \cup \cdots D_V$.
\end{defi}


Given that ${\bf A}_f^n$ is a block matrix, we next define GCT for block matrices.

\begin{defi}[GCT for block matrices \cite{tretter2008spectral, salas1999gershgorin}]
\label{def-gresh-block} 
Consider a symmetric matrix ${\bf A} \in \mathbb{R}^{mn \times mn}$, composed of $n$ block matrices, where each block matrix is denoted by ${\bf A}_{kt} \in \mathbb{R}^{m \times m}$ for $k,t \in [m]$. Let $\Xi({\bf A})$ represent the set of all eigenvalues of ${\bf A}$. The circle corresponding to eigenvalue $\lambda_k$ of the block matrix ${\bf A}$ is then defined as follows:
\begin{align}
\label{eq-gresh-block-matrix}
D^b_k = \{ \: \lambda_k({\bf A}) \in \Xi({\bf A}) \: :  \,  \, \left| \lambda_k({\bf A}) - \lambda_l({\bf A}_{kk}) \right| \leq \sum_{t=1, \, t \neq k}^{n} |{\bf A}_{kt}| \},
\end{align}
where $l\in [n]$, and the set of eigenvalues of ${\bf A}$, i.e., $\Xi({\bf A})$, satisfy
\begin{align}
\label{eq-block-gresh-disc-union}
    \Xi({\bf A})\in \cup_{k=1}^{n}D^b_{k} \ .
\end{align}
\end{defi}

From (\ref{eq-gresh-block-matrix}), we deduce that the regions (circles) covering the eigenvalues of ${\bf A}$ are centered at eigenvalues of ${\bf A}_{kk}$, and circles' radius are enlarged by the size of the non-diagonal matrices ${\bf A}_{kt}$. Given an $n$-fold OR product graph with an adjacency matrix ${\bf A}_f^n$, where ${\bf A}_f^n$ is a ${V^n\times V^n}$ binary and symmetric matrix, the circles $D^b_{k}$ in (\ref{eq-gresh-block-matrix}) can be simplified to block intervals $\delta^b_k$.



\section{Bounds On Cyclic and Regular Graphs}
\label{sec:results}

We here detail characteristic graphs that are $d$-regular and derive lower and upper bounds on their chromatic numbers and graph entropies. Our novel contributions include the characterization of chromatic numbers for $n$-fold OR products of cycles, as well as a novel coloring scheme for OR products of odd cycles, as detailed in Section~\ref{sec:coloring_cyclic_graphs}. 
In Section~\ref{entropy-cycles-and-general-graph}, we bound the entropy of characteristic graphs for cycles. 
Given a cyclic graph, in Section~\ref{sec:Eig-Adj-Mat-Cycl-Graph-}, we analyze its key properties using the eigenvalues of its adjacency matrix, and in Section~\ref{sec:choromatic-using-eigenvalue-cycles}, we bound its chromatic number. 
In Section~\ref{sec:sub-d-reg}, we characterize the degrees and chromatic numbers for regular graphs and their OR products. Section~\ref{sec:graph-expansion} details the expansion rates of OR products of regular graphs, 
with implications on the fundamental limits of compressibility of such graphs.


\subsection{Coloring Cyclic Graphs}
\label{sec:coloring_cyclic_graphs}

Let $C_i$ be a cycle graph with $i$ vertices, that represents the characteristic graph that source $X_1$ builds for computing $f$ (similarly for source $X_2$). For an even cycle $i=2k$, and for an odd cycle, $i=2k+1$, for some $k\in \mathbb{Z}^+$. 
We seek to compress $G_{X_1}$ and $G_{X_2}$ to recover the desired function outcome at a receiver in an asymptotically lossless manner. 
To that end, we will determine the minimum entropy coloring for the $n$-fold OR product of $C_i$, denoted by $C_i^n$ (and similarly for $G_{X_{2}}$), for the receiver to decode $f$ from the received colors. 

We start by determining the degree of each vertex in $C_{i}^{j}$ for $j\in[n]$. 

\begin{prop}[Degree of vertices in $C_i^n$]
 \label{prop-deg-cycles}
The degrees\footnote{In regular graphs, where all vertices have the same degree, we omit the index $k$ of $x_k$ in Propositions~\ref{prop-deg-cycles} and~\ref{prop-deg-regular}.} in the $n$-fold OR product of a cycle graph, $C_i^n$ for $n\geq 2$, are calculated as follows:
\begin{align}
\label{eq-degree_function}
deg({x}^n)= 2+\sum_{j=1}^{n-1}2 (V^j )=2\cdot\frac{V^n-1}{V-1} \ , \quad \forall {x}^{n}\in [V^n] \ . 
\end{align}
\end{prop}
\begin{proof}
See Appendix \ref{App:theo-deg-even-cycles}.
\end{proof}

From Proposition~\ref{prop-deg-cycles}, we infer that for a given pair of vertices ${x}_t, {x}_k \in [V]$ where $t\neq k$, if $deg({x}_t)=deg({x}_k)$, then for the $n$-fold OR product, $deg({x}_t^n)=deg({x}_k^n)$, for ${x}_t^n, {x}_k^n \in [V^n]$, i.e., taking the $n$-fold OR products does not alter the equality of degrees. 
Therefore, for any $d$-regular graph (including cycles), we derive the following result about the regularity of its OR products.

\begin{cor}
\label{cor-total-edge-d-regular}
Given a $d$-regular graph $G_{d, V}$, where $d=deg({x})$, its $n$-fold OR product with itself for $n\geq 1$, i.e., $G_{d,V}^n$, is also a regular graph, with degree $deg({x}^n)$, and total number of edges $E^n=\Big(\sum_{k=1}^{V^n}deg({x}^{n})\Big)/2$. 
\end{cor}

\subsubsection{Even Cycles}
\label{sec-even}
 Here, we consider even cycles, denoted by $C_{2k}$, $k\in\mathbb{Z}^+$. Vertices of $C_{2k}$ are sequentially numbered clockwise from $0$ to $2k-1$ (e.g., see $G_{X_1}$ in Figure~\ref{fig:funcom}), and alternatingly colored. Vertices with even indices are assigned one color, while those with odd indices receive another. We next determine the chromatic number of $C_{2k}^n$, denoted as $\chi(C_{2k}^{n})$. 
  
\begin{prop}[Chromatic number of $C_{2k}^{n}$]
\label{prop-even-cycle-chorom}
The chromatic number of $C_{2k}^{n}$ is given as  
\begin{align}
\label{eq-opec-even-chrom}
\chi ({C_{2k}^{n}})=2^n \ , \quad k\in\mathbb{Z}^+ \ ,\quad n\geq 1 \ .    
\end{align}
\end{prop}
\begin{proof}
Given $C_{2k}$, with $\chi(C_{2k})=2$, its $2$-fold OR product $C_{2k}^{2}$ consists of $(2k)^2$ vertices and $2k$ sub-graphs, $\{C_{2k}^2(1),C_{2k}^2(2), \dots, C_{2k}^2(2k)\}$, each containing $2k$ vertices. Since each sub-graph is two-colorable and fully connected to its neighbors, adjacent sub-graphs must use different colors. For instance, $\{C_{2k}^2(1),C_{2k}^2(2)\}$ requires 4 colors in total. However, due to the cyclic structure of OR products, alternating colors from $C_{2k}^2(1)$ can cover $C_{2k}^2(3)$, and similarly for odd-indexed sub-graphs, meaning that $\chi(C_{2k}^2)=4$. This method can also calculate $\chi ({C_{2k}^n}$) from $(n-1)$-fold to $n$-fold OR products. Figure~\ref{fig:3rdPower_C4} shows a valid coloring for $C_4^3$, where $\chi ({C_{4}^{3}})=8$. Similarly, by induction, $\chi ({C_{2k}^{n}})$ satisfies (\ref{eq-opec-even-chrom}).
\end{proof}
\begin{figure}[htbp]
\centerline{\includegraphics
[width=0.5\linewidth, height=40mm]{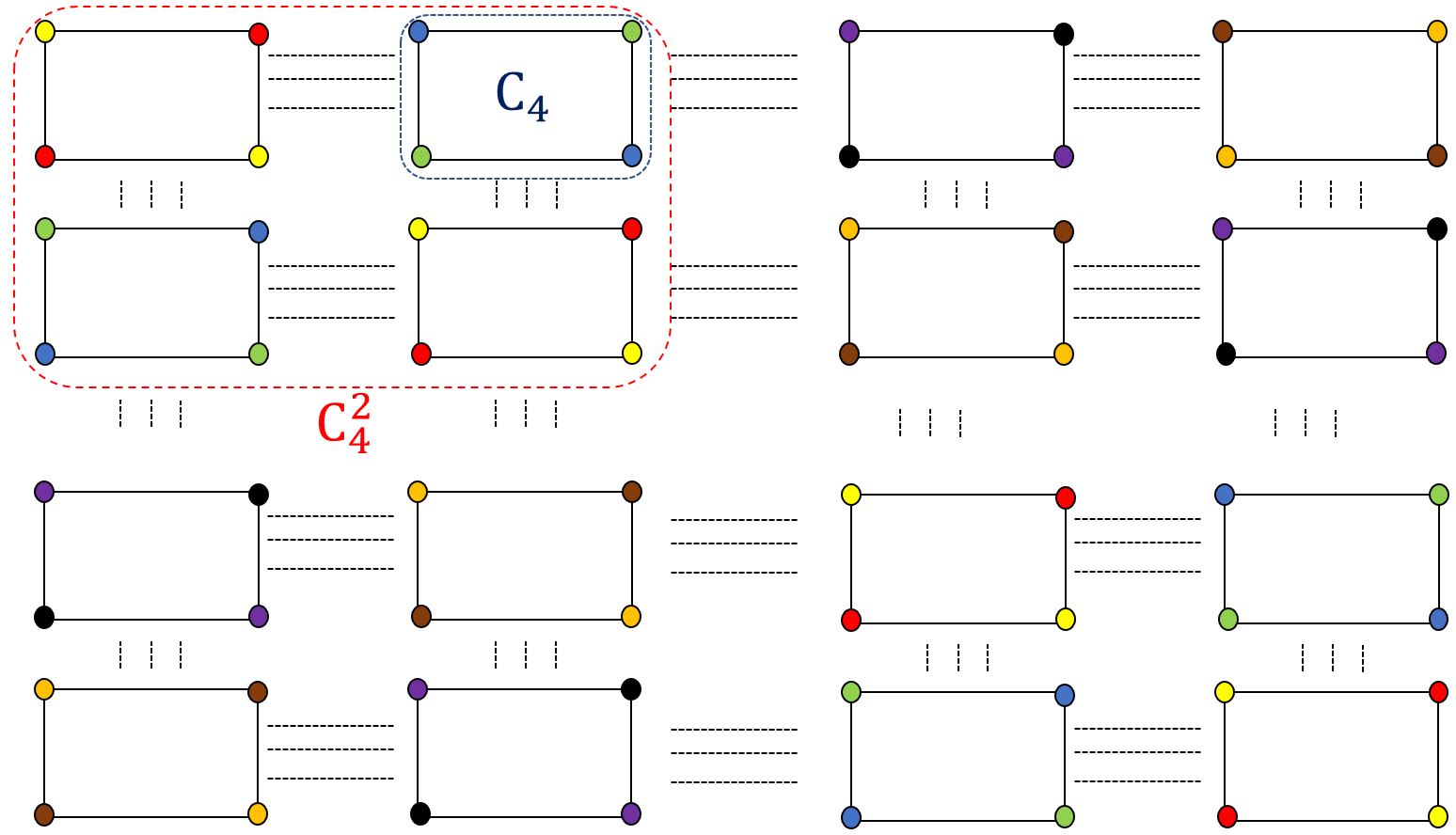}}
\caption{A valid coloring of $C_{4}^{3}$ with 8 colors.}
\label{fig:3rdPower_C4}
\end{figure}

\subsubsection{Odd Cycles}
\label{sec-odd}

We here focus on odd cycles, namely $C_{2k+1}$, $k\in \mathbb{Z}^+$, and their $n$-fold OR products. In the special case with $3$ vertices, $C_{3}$ is a complete graph, and a valid coloring requires $3$ distinct colors for a receiver to successfully recover the function.
Furthermore, for a valid coloring of $C_3^n$, $\chi({C_3^n})= 3^n$ for $n \geq 1$. For coloring of an odd cycle with the length $i=2k+1$, for $k\geq 2$, one could reuse the colors. For instance, given $C_5$, we have $\chi({C_5})=3$. We next present an achievable scheme for valid colorings of general odd cycles. 

 \begin{prop}[Chromatic numbers of odd cycles]
\label{prop-chromatic_number_cycles-odd}
The chromatic number of $C_i^{n+1}$, denoted as $\chi ({C_i^{n+1}})$, can be recursively computed from $\chi({C_i^{n}})$ as follows:
\begin{align}
\label{eq-chromatic_number_cycles}
 \chi ({C_{i}^{n+1}})= 2\chi ({C_{i}^{n}})  +\left\lceil\frac{\chi({C_{i}^{n}})}{2}\right\rceil \quad ,\ i=2k+1 \text{ and } k\in\mathbb{Z}^{\geq2}.
  \end{align}
 \end{prop}

\begin{proof}
See Appendix \ref{App:chromatic_number_cycles}.
\end{proof}

For even cycles, from Proposition~\ref{prop-even-cycle-chorom}, $\chi({C_{i}^{n}})=2^n$. For odd cycles, in Section~\ref{sec:choromatic-using-eigenvalue-cycles}, we will establish upper and lower bounds on $\chi (C_i^n)$ using the adjacency matrix of $C_i^n$, denoted as ${\bf A}_f^n$.

\begin{figure}[htbp]
\centerline{\includegraphics[scale=0.19]{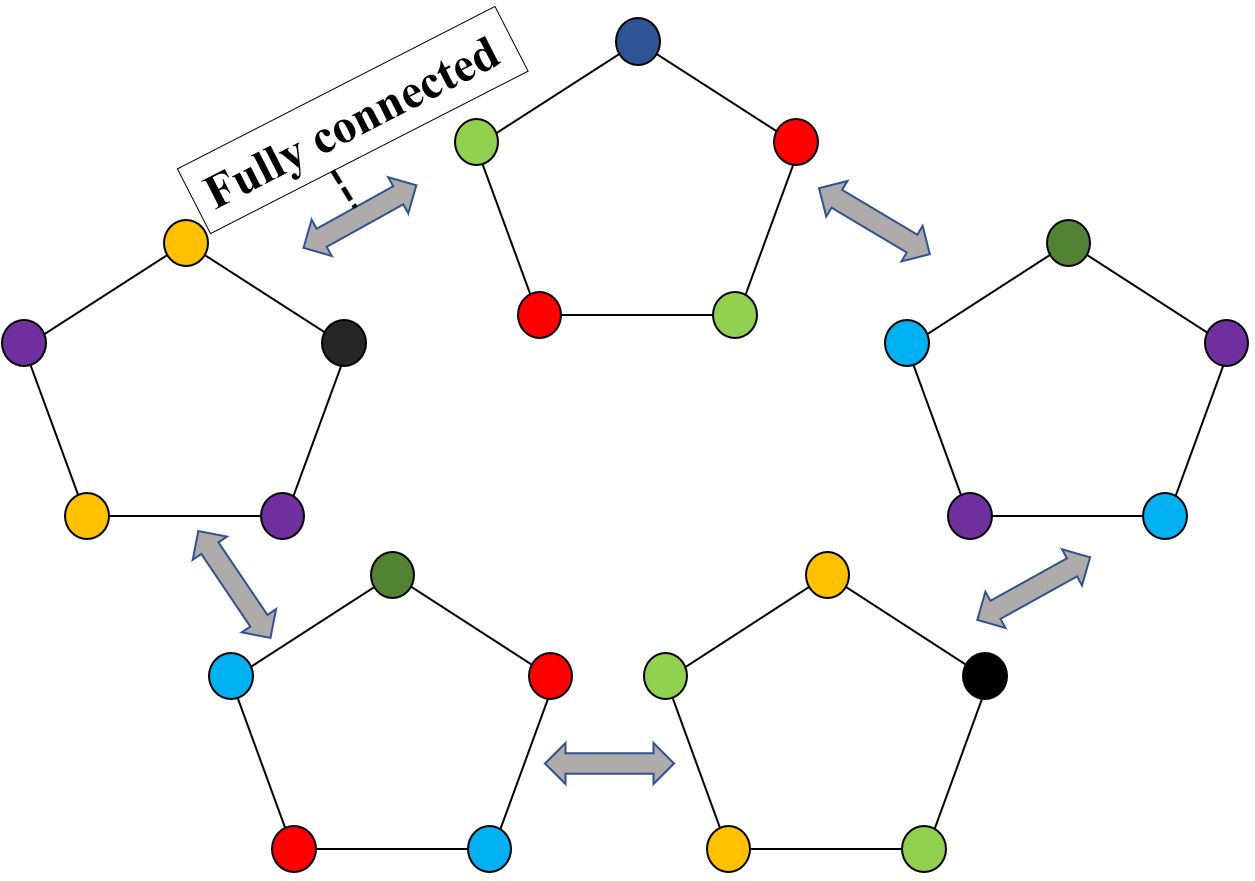}}
\caption{The $2$-fold OR product of $C_5$, i.e., $C_5^2$, and its valid coloring.}
\label{fig:2nd_power_C5}
\end{figure}

We next demonstrate the gain in terms of the required number of colors of our approach in Proposition~\ref{prop-chromatic_number_cycles-odd} over a greedy algorithm that does not leverage the structure of $C_{2k+1}^n$ in coloring.

\begin{prop}[The multiplicative gain of our approach over a greedy coloring algorithm]
\label{prop_gain_over_greedy} 
The gain of the recursive coloring approach in Proposition~\ref{prop-chromatic_number_cycles-odd} for $C_{2k+1}^n$, $k\in\mathbb{Z}^{\geq2}$, over the greedy algorithm, which calculates $\chi({C_{2k+1}})$ and uses $(\chi({C_{2k+1}}))^n$ colors for coloring $C_{2k+1}^n$, is
\begin{align}
\label{eq-gain_vs_greedy} 
\eta_n = \frac{\big(\chi({C_{2k+1}})\big)^n}{\chi({C_{2k+1}^n)} } \geq 1.2^n  \ ,
\end{align}
which is exponential and unbounded as $n\to\infty$, i.e., $\eta=\lim_{n\to \infty}\eta_n=\infty$.
\end{prop}

\begin{proof}
\label{Proof_prop_gain_over_greedy} 
See Appendix~\ref{App:prop_gain_over_greedy}.
\end{proof}
\begin{figure}[h] 
    \centering
\includegraphics[width=0.40\linewidth]{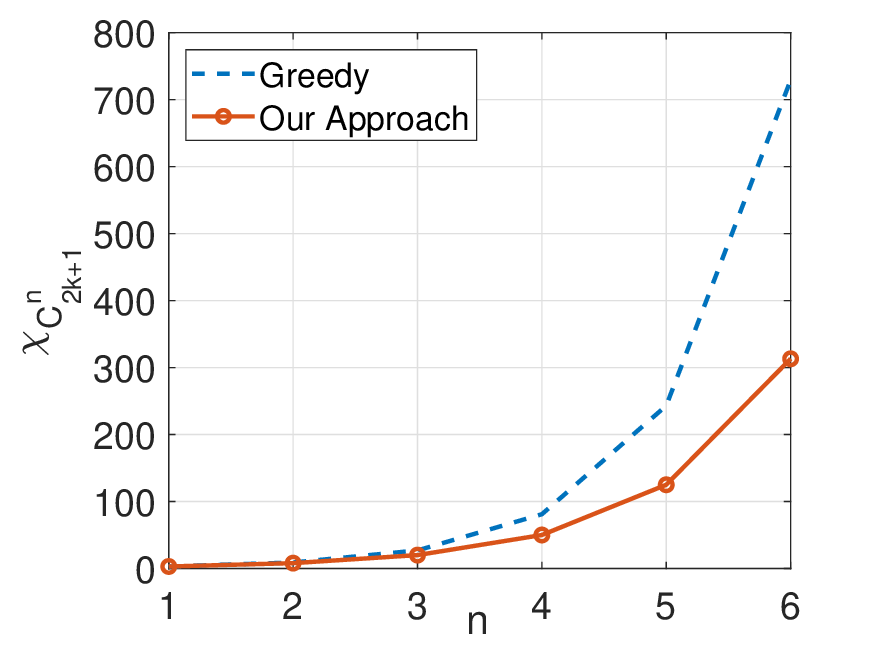}
\includegraphics[width=0.42\linewidth]{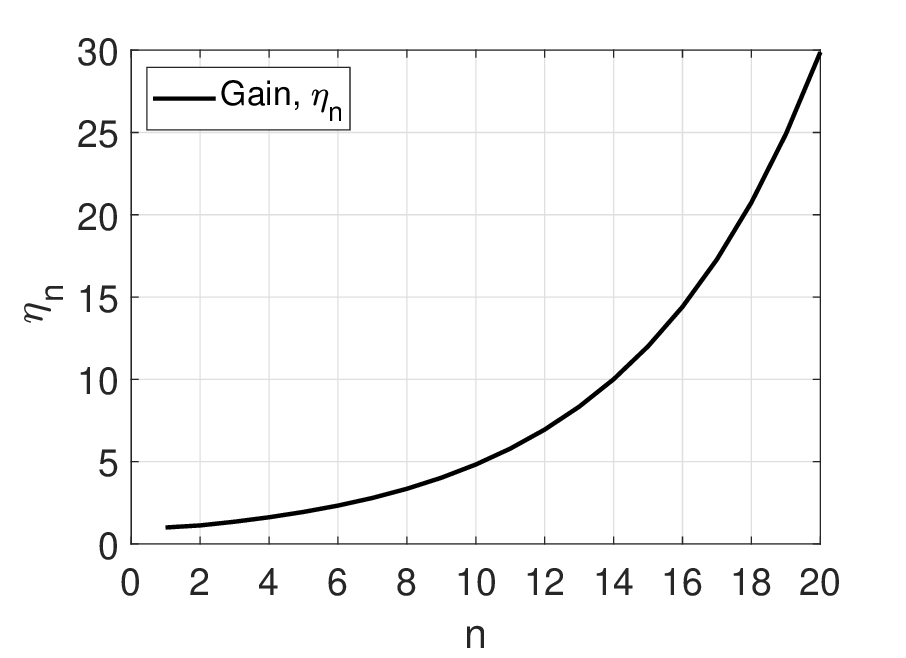}
\caption{(Left) $\big(\chi({C_{2k+1}})\big)^n$ (dashed (blue) curve), and $\chi({C_{2k+1}^n})$ (solid (orange) curve) for any $k\geq 2$. (Right) The gain, i.e., $\eta_n$, of the coloring approach in Proposition~\ref{prop-chromatic_number_cycles-odd} compared to the Greedy algorithm.} 
\label{fig:the_gain_vs_greedy}
\end{figure}


\subsection{Bounding the Chromatic Entropy of Cycles}
\label{entropy-cycles-and-general-graph}
Here, we establish an upper bound on the chromatic entropy of $C_{i}^n$, for $i\in\mathbb{Z}^+$.

\subsubsection{Entropy of an even cycle}
\label{sub-entropy of evens}

From Proposition~\ref{prop-even-cycle-chorom}, $\chi({C_{2k}^{n}})=2^n$. We recall that the chromatic entropy $H^{\chi}_{C_{2k}^n}({\bf X}_1)$, which is the minimum achievable entropy of a valid coloring of $C_{2k}^n$ \cite{AO96, korner1973coding}, and the characteristic graph entropy $H_{C_{2k}}({X}_1)$ are related as follows~\cite[Theorem~5]{AO96}:
\begin{align}
\label{eq-chor-ent}
    H_{C_{2k}}({X}_1)=\lim_{n\to \infty} \frac{1}{n}H^\chi_{C_{2k}^n}({\bf X}_1) \ ,
\end{align}
where 
\begin{align}
\label{eq-chorm_korner-cycle}
    H^{\chi}_{C_{2k}^n}({\bf X}_1)= \min_{\mathcal{C}_{C_{2k}^n}} H(\mathcal{C}_{C_{2k}^n}) \ .
\end{align}

If the distribution of $C_{2k}^n$ is uniform, then $H_{C_{2k}}({X}_1)$ is given as follows:
\begin{align}
\label{eq-note_even-entropy}
H_{C_{2k}}({X}_1)=\lim_{n\to \infty} \frac{1}{n}\log_2 2^n=1 \quad \mbox{bits} \ .
\end{align}    

\subsubsection{Entropy of an odd cycle}
\label{entropy of odds}
We next examine the chromatic and graph entropies of $C_{2k+1}$. Unlike $C_{2k}$, the coloring PMF for odd cycles, $p(\mathcal{C}_{C_{2k+1}^{n}})$, is non-uniform (i.e., $H_{C^n_{2k+1}}^{\chi}({\bf X}_1) < \log_2 \chi(C^n_{2k+1})$), as demonstrated through the examples below.

\begin{exm}
\label{ex_entropy_second}
Given that $X_1$ follows a uniform distribution over an alphabet $\mathcal{X}_1$ such that $|\mathcal{X}_1|=5$ and with a characteristic graph $G_{X_1}=C_5$, the coloring PMF satisfies, $p(\mathcal{C}_{C_5}) = \{\frac{1}{5}, \frac{2}{5}, \frac{2}{5}\}$ with $\chi({C_5})=3$, and using (\ref{eq-chorm_korner-cycle}) for chromatic entropy, yields $H_{C_5}^{\chi}(X_1)=\min_{\mathcal{C}_{C_5}} H(\mathcal{C}_{C_5})=1.52$. For the $2$-fold OR product graph $C_5^2$, where the PMF that satisfies the minimum coloring entropy is $p(\mathcal{C}_{C_5^2}) = \{\frac{4}{25}, \frac{4}{25}, \frac{4}{25}, \frac{4}{25}, \frac{4}{25}, \frac{2}{25}, \frac{2}{25}, \frac{1}{25}\}$ with $\chi({C_5^2})=8$, similarly using (\ref{eq-chorm_korner-cycle}) we obtain
\begin{align}
\label{exm_eq_2nd_ent_odd}
    \frac{1}{2}H(\mathcal{C}_{C_5^2}) =1.37 < H(\mathcal{C}_{C_5})=1.52 \quad \text{bits}\ .
\end{align}  
\end{exm}
While from  Example~\ref{ex_entropy_second}, we can determine $H^{\chi}_{C_5}({X}_1)$ and $H^{\chi}_{C_5^2}({\bf X}_1)$, determining $H^\chi_{C_5^n}({\bf X}_1)$, which corresponds to the minimum entropy among all possible valid colorings, becomes complex for large $n$. Next, given a characteristic graph $G_{X_1}=C_{2k+1}$, we establish an upper bound on $H_{C_{2k+1}}(X_1)$ by employing (\ref{eq-chor-ent}) and (\ref{eq-chorm_korner-cycle}) and devising valid colorings for the MISs of $C_{2k+1}^n$.




\begin{prop}[An upper bound on $H_{C_{2k+1}}({X}_1)$]
\label{prop-upper_entrop-cycle}
The entropy is upper bounded as follows:
\begin{align}
\label{eq-coroll-upper_entrop}
  H_{C_{2k+1}}({X}_1)\leq  \frac{1}{n}  H\Big(\alpha_{n}\cdot \Big(\frac{k^n}{(2k+1)^n}\Big), \alpha_{n-1}\cdot \Big(\frac{k^{(n-1)}}{(2k+1)^{n}}\Big), \cdots, \alpha_0 \cdot\Big( \frac{1}{(2k+1)^n}\Big) \Big)  \   ,
\end{align}
where $\alpha_t\ , \ t\in [n]\cup \{0\}$ represents the number of MISs with $|\text{MIS}_{C^t_{2k+1}}|$, and $\alpha_n$ satisfies
\begin{align}
\label{eq-coroll-upper_entrop_alpha_n}
     \frac{(2k+1)^n\cdot (k^2-1)}{k^{2(n+1)}-1}\cdot k^n  <\alpha_n < \frac{(k^{2(n+1)}-1)\cdot(2k+1)^n\cdot k^{(n-1)}}{(2k+1)^{2n}\cdot(k^2-1)\cdot k^{(n-1)}-(k^{2(n+1)}-1)}  \ .
\end{align}
\end{prop}

\begin{proof}
See Appendix \ref{App:corr_Entropy_bound_odd}.
\end{proof}

In (\ref{eq-coroll-upper_entrop}), as color reuse increases (i.e., independent sets with high cardinality and high $\alpha_n$), $H_{C_{2k+1}}({X}_1)$ decreases. We later generalize Proposition~\ref{prop-upper_entrop-cycle} to general graphs in Section~\ref{sec:entropy_general_graphs}. 

Next, we apply Proposition~\ref{prop-upper_entrop-cycle} to $C_5^3$ to derive an upper bound to $H(\mathcal{C}_{C_5^3})$.

\begin{exm}
 \label{example_entropy_bound}
The cardinalities of the MISs for different graph products $C_5^j$, where $j\in [3]\cup \{0\}$, are given as $|\text{MIS}_{C_5^0}|=2^0,\  |\text{MIS}_{C_5^1}|=2^1, \  |\text{MIS}_{C_5^2}|=2^2$, and $|\text{MIS}_{C_5^3}|=2^3$, respectively, where $\chi({C_{5}^3})=20$, which can be recursively computed using (\ref{eq-chromatic_number_cycles}). 
Employing Proposition~\ref{prop-upper_entrop-cycle} yields 
\begin{align}
\label{eq-exm-eq-c_5}
    \alpha_3\cdot\Big(\frac{8}{125}\Big)+  \alpha_2\cdot\Big(\frac{4}{125}\Big)+   \alpha_1\cdot\Big(\frac{2}{125}\Big)+   \alpha_0\cdot\Big(\frac{1}{125}\Big)=1 \ .
\end{align}

Employing the ordering $\alpha_3>\alpha_2>\alpha_1>\alpha_0$ in (\ref{eq-exm-eq-c_5}) leads to the simplification, $64 \alpha_0+ 16 \alpha_0+ 4\alpha_0+ \alpha_0 \leq 125$. Hence, $\alpha_0\leq \frac{25}{17}$. Because $\alpha_t\in\mathbb{Z}^+$, $t\in [3] \cup \{0\}$, we have $\alpha_0=1$. Employing the same ordering in (\ref{eq-exm-eq-c_5})
also yields the simplification $8\alpha_3+ 4(\frac{\alpha_3}{2})+2(\frac{\alpha_3}{4})+ \frac{\alpha_3}{8} \geq 125$, which yields the condition $\alpha_3\geq \frac{200}{17}$.
Because $\alpha_t\in\mathbb{Z}^+$, (\ref{eq-exm-eq-c_5}) yields the upper bound $\alpha_3\leq  15$. Employing these lower and upper bounds to (\ref{eq-coroll-upper_entrop}) leads to $1.37 \leq \frac{1}{3} H_{C_5^3}^{\chi}({\bf X}_1)\leq 1.41$.
\end{exm}

Next, we consider another example for $C_5^3$, where we employ Proposition~\ref{prop-upper_entrop-cycle} to determine $\alpha_t$'s and $p(\mathcal{C}_{C^3_5})$ (similar to Examples~\ref{ex_entropy_second}-\ref{example_entropy_bound}). We then apply Huffman coding~\cite{huffman1952method} to the coloring random variable to optimize the compression rate of $C_5^3$.


\begin{exm}[Huffman coding for a given characteristic graph coloring] \label{ex-huffman-MIS-color-dist}
Consider the setting of Example~\ref{ex_entropy_second}, where $\chi({C_5})=3$. MISs are represented by the colors $Y$, $B$, and $R$ with PMF $p(\mathcal{C}_{C_5})= \left\{\frac{1}{5}, \frac{2}{5}, \frac{2}{5}\right\}$. To achieve $H_{C_5}^{\chi}(X_1)$ approximately, a binary Huffman tree is constructed for each color. The assigned codes are $Y: 1$, $R: 00$, and $B: 01$. Similarly, for $C_5^2$, from (\ref{eq-chromatic_number_cycles}), $\chi({C_5^2})=8$. The color set is denoted as $\mathcal{C}({C_5^3})= \{c_1, c_2, \dots, c_8\}$, with the corresponding Huffman codes: $c_1: 11$, $c_2: 000$, $c_3: 001$, $c_4: 010$, $c_5: 011$, $c_6: 101$, $c_7: 1000$, and $c_8: 1001$. For $C_5^3$, $\chi({C_5^3})=20$, we have the following coloring PMF: 
\begin{align}
\label{eq:PMF_C_5_3}
p(\mathcal{C}_{C_5^3})= \left\{\frac{8}{125},\cdots ,\frac{8}{125}, \frac{4}{125}, \cdots, \frac{4}{125}, \frac{2}{5},\frac{2}{5},\frac{1}{5}\right\},
\end{align}
where there are $\alpha_3=13$ colors with probability $\frac{8}{125}$, $\alpha_2=4$ colors with probability $\frac{4}{125}$, $\alpha_1=2$ colors with probability $\frac{2}{125}$, and $\alpha_0=1$ color with probability $\frac{1}{125}$, where the set of $\alpha_t$'s are uniquely determined employing $\chi({C_5^3})=\sum_{t=0}^{3}\alpha_t= 20$ and $\alpha_t\geq \alpha_{t-1}$. Building the binary Huffman encoding tree using $p(\mathcal{C}_{C_5^3})$ in (\ref{eq:PMF_C_5_3}), helps assign codes for encoding $\mathcal{C}_{C_5^3}$, that achieves the minimum average code length. 
\end{exm}

Next, we examine the relationship between the chromatic numbers of cycles and the eigenvalues of their adjacency matrices.

\subsection{Eigenvalues of the Adjacency Matrices of \texorpdfstring{$C_i$}{}}
\label{sec:Eig-Adj-Mat-Cycl-Graph-}

Let ${\bf A}_f \in \{0,1\}^{V \times V}$ be the adjacency matrix of $G(\mathcal{V}, \mathcal{E})$, where ${\bf A}_f(l_1, l_2) = 1$ if distinct vertices $l_1, l_2 \in [V]$ must be distinguished, and $0$ otherwise. Since there are no self-loops, ${\bf A}_f(l_1, l_1) = 0$. The adjacency matrix of the $2$-fold OR product, ${\bf A}_f^2$, is composed of diagonal blocks of ${\bf A}_f$, where entries equal to 1 are replaced with all-one matrices ${\bf J}_V$, and zeros with all-zero matrices ${\bf Z}_V$, representing full or no connectivity between sub-graphs $C_i^2(l)$ for $l \in [V]$.
 
Similarly, 
we can construct ${\bf A}_f^n$ for the $n$-fold OR product using induction. Given the $n$-fold OR product of $C_i$ with itself, the adjacency matrix of $C_i^n$, 
has the following block structure: 
\begin{align}
\label{eq-AdjacencyMatrix}
{\bf A}_{f}^n=
\begin{bmatrix}
{\bf A}_f^{n-1} & {\bf J}_{V^{n-1}} & {\bf Z}_{V^{n-1}} & \dots & {\bf Z}_{V^{n-1}} & {\bf J}_{V^{n-1}}\\
{\bf J}_{V^{n-1}} & {\bf A}_f^{n-1} & {\bf J}_{V^{n-1}} & {\bf Z}_{V^{n-1}} & \dots & {\bf Z}_{V^{n-1}}\\
\vdots & \vdots & \vdots & \ddots & \vdots & \vdots\\
{\bf J}_{V^{n-1}} & {\bf Z}_{V^{n-1}} & {\bf Z}_{V^{n-1}} & \dots & {\bf J}_{V^{n-1}} & {\bf A}_f^{n-1} \  
\end{bmatrix}\ ,
\end{align}
consisting of $V$ row-block matrix partitions of size $V^{n-1}$ each. In every block row of ${\bf A}_{f}^n$, there are exactly two ${\bf J}_{V^{n-1}}$ matrices and one ${\bf A}_f^{n-1}$, i.e., the adjacency matrix of the $(n-1)$-fold OR product of $C_i$. We next characterize the eigenvalues of the all-ones matrix ${\bf J}_V$ of size $V \times V$, which represents full connectivity between adjacent sub-graphs in the $2$-fold OR product (see Definition~\ref{OR power}). The proof of the following lemma is given in \cite[Lemma~1]{10313467}.
\begin{lem}
\label{eigenvalue_dependency}
The eigenvalues of the all-ones matrix ${\bf J}_V\in 1^{V\times V}$ are $0$ and $V$, with algebraic multiplicities $V-1$ and $1$, respectively.
\end{lem}

From Lemma \ref{eigenvalue_dependency}, we have 
$\lambda_1({\bf J}_{V})=V$. To calculate the eigenvalues of ${\bf A}_{f}^n$, one needs to solve $|{\bf A}_{f}^{n}- \lambda({\bf A}_f^n) {\bf I}_{V^{n}}|=0$. Let $\{\lambda_k({\bf A}_f), \ k\in[V]\}$ be the set of eigenvalues of ${\bf A}_f$, and $\{\nu_k({\bf A}_f^2), \ k\in[V^2]\}$ be the set of eigenvalues of the $V^2\times V^2$ matrix ${\bf A}_f^2$. We also let $\{{\bf u}_k, \ k\in[V]\}$ and $\{{\bf v}_k, \ k\in[V^2]\}$ be the sets of eigenvectors of ${\bf A}_f$ and ${\bf A}_f^2$, respectively. Solving ${\bf A}_f {\bf u}=\lambda_k({\bf A}_f) {\bf u}$ determines $\lambda_k({\bf A}_f)$ for $k\in[V]$, associated with ${\bf u}$. Similarly, for the $2$-fold OR product graph, the eigenvalues of ${\bf A}_f^2$ are determined by solving
\begin{align}
\label{Af2_eigenvalues}
{\bf A}_f^2{\bf v}=\nu({\bf A}_f^2) {\bf v}\ ,
\end{align}
where ${\bf v}=[{\bf v}^{\intercal}_1,{\bf v}^{\intercal}_2, \dots, {\bf v}^{\intercal}_V]^{\intercal}$, and ${\bf v}_k$ is a  $V\times 1$ vector, for each $k\in[V]$. In other words, for each $\nu({\bf A}_f^2)$, we have a set of $V$ block row equations, each containing $V$ scalar equations. More specifically, $\nu({\bf A}_f^2)$ in (\ref{Af2_eigenvalues}) 
satisfies the following $V$ block equations:
\begin{align}
\label{eq-row-cycle-matrix-id}
{\bf A}_f {\bf v}_k+ {\bf J}_{V} {\bf v}_{k+1}+ {\bf J}_{V} {\bf v}_{k-1}=\nu({\bf A}_f^2)  {\bf v}_k\ ,\quad k\in[V]\ .
\end{align}

Using (\ref{eq-row-cycle-matrix-id}), we next derive the eigenvalues of ${\bf A}_f^{n}$.

\begin{theo}[Distinct eigenvalues of $C_i^n$]
\label{cyclic_adjacency}
The adjacency matrix ${\bf A}_f^n$ for $C_i^n$ where $n\ge 2$, has the same distinct eigenvalues of ${\bf A}_{f}^{n-1}$ as well as two new distinct eigenvalues.  
\end{theo}
\begin{proof}
See Appendix~\ref{App:cyclic_adjacency}.
\end{proof}
Next, in Section~\ref{sec:choromatic-using-eigenvalue-cycles}, we will examine the relation between the eigenvalues and the chromatic number of $C_i^n$ to establish lower and upper bounds on $\chi({C_i^n})$.


\subsection{Bounding the Chromatic Number of \texorpdfstring{$C_i$}{} Using the Eigenvalues of \texorpdfstring{${\bf A}_f$}{}}
\label{sec:choromatic-using-eigenvalue-cycles}    

Given a cycle $C_i$ with an adjacency matrix ${\bf A}_f$, where we denote by $\vartheta(C_i)$ the set of its distinct eigenvalues, $\lambda_1({\bf A}_f)$ its largest, and $\lambda_{V}({\bf A}_f)$ its smallest eigenvalue \cite{aspvall1984graph}. Exploiting these, we can derive the following lower and upper bounds on $\chi({C_i})$ as follows \cite{hoffman1970eigenvalues, wilf1967eigenvalues}:
\begin{align}
\label{eq-chro_up-low_eig}
 1-\frac{\lambda_1({\bf A}_f)}{\lambda_{V}({\bf A}_f)}\leq \chi({C_i})\leq \lfloor \lambda_1({\bf A}_f) \rfloor +1\ ,
 \end{align}
where two lower bounds on $\lambda_V$ have been derived in \cite{brigham1984bounds} and \cite{hong1988bounds}, respectively:
\begin{align}
\label{birgham-low-bound-small-eigen}
\lambda_{V}({\bf A}_f)&\ge -\sqrt{2E\cdot((V-1)/{2})} \ ,
    \\
\label{hong-low-bound-small-eigen}
\lambda_{V}({\bf A}_f)&\ge -\sqrt{{V}/{2}\cdot({(V+1)}/{2})} \ .  
\end{align}  
 
We next apply the bound in (\ref{eq-chro_up-low_eig}) to $\chi({C_i^j})$ corresponding to ${\bf A}_f^{j}$, $j\in[n]$. To that end, let us consider an example where $i=5$ and $j=2$. Given $C_5$, the set of distinct eigenvalues of ${\bf A}_f$, using numerical evaluation, are given as $\vartheta(C_5)=\{-1.618, 0.618, 2\}$. 
Note also that $\vartheta(C_{5}^{2})=\{-6.09,-1.61803,$  $0.61803, 5.09016,12\}$ for ${\bf A}_{f}^{2}$. Consider the set $\vartheta(C_{5}^{2})$, where $\lambda_{V^2}({\bf A}_f^{2})=-6.09$, for the $2$-fold OR product graph $C_{5}^{2}$, an application of (\ref{birgham-low-bound-small-eigen}) yields $\lambda_{V^2}({\bf A}_f^{2})\ge-60$. Similarly, an application of (\ref{hong-low-bound-small-eigen}) yields that $\lambda_{V^2}({\bf A}_f^2)\geq -12.748$. Hence, for the $2$-fold OR product, one can find that the bound in \cite{hong1988bounds} is tighter for cycles and their $n$-fold OR products. 

We do not have an exact characterization for $\lambda_{V^n}({\bf A}_f^n)$. However, using the lower bounds in (\ref{birgham-low-bound-small-eigen}) and (\ref{hong-low-bound-small-eigen}) can help derive a lower bound on $\chi({C_i^n})$. 
On the other hand, for $\lambda_{1}({\bf A}_f^n)$, recalling from Proposition~\ref{prop-deg-cycles} that all vertices of $C_i^n$ for $n\geq 2$ have equal degrees, and exploiting this feature, we next derive an exact characterization of $\lambda_1({\bf A}_f^n)$, for $n\geq 2$. 

\begin{prop}[The largest eigenvalue for the adjacency matrix of $C_i^n$]
\label{prop-largest-eigenvalue}
The largest eigenvalue of $C_{i}$ is $\lambda_1({\bf A}_f)=2$, and the largest eigenvalue of the $n$-fold OR product $C_{i}^{n}$ is determined as
\begin{align}
\label{eq:lambda_1_C_i_n_power}
\lambda_{1}({\bf A}_f^{n})=2+\sum_{j\in [n-1]}(2V^j) \ ,\quad n\geq 2 \ .
\end{align}
\end{prop}

\begin{proof}
\label{proof-largest-eig-cyc}
We prove it using Definition~\ref{def-Degree-of-a-vertex}, (\ref{eq-chro_up-low_eig}), and the proof of Theorem~\ref{cyclic_adjacency}—which leverages the fact that the eigenvalues of the sum of adjacency matrices equal the sum of their individual eigenvalues. The calculation of eigenvalues for the adjacency matrix of $C_{i}^{j}$ is given in (\ref{eq-row-cycle-matrix-id}), which links the eigenvalues of ${\bf A}_f^j$ with the two ${\bf J}_V$ matrices of size $V^{j-1} \times V^{j-1}$ in each block row. From (\ref{eq:eigen_larg}) and (\ref{adjacency_cyclic_relation}), the largest eigenvalue for a power graph $C_i^n$ is obtained by adding $\lambda_1({\bf A}_f^{n-1})$ for a sub-graph $C_i^{n-1}$ and twice the largest eigenvalue of ${\bf J}_{V^{n-1}}$, which is $V^{n-1}$ (see Lemma~\ref{eigenvalue_dependency}). Thus, $\lambda_1({\bf A}_f) =2$, and for $n \geq 2$, we achieve (\ref{eq:lambda_1_C_i_n_power}).  
\end{proof}

Next, we present new bounds on $\chi({C_{i}^j})$ by exploiting (\ref{eq-chro_up-low_eig}), (\ref{birgham-low-bound-small-eigen}) and (\ref{hong-low-bound-small-eigen}) which lower bound $\lambda_{V^j}({\bf A}_f^{j})$, and Proposition~\ref{prop-largest-eigenvalue}, which gives the
exact value of $\lambda_1({\bf A}_f^{j})$.

\begin{prop}[Bounding $\chi({C_{i}^{n}})$ using eigenvalues of ${\bf A}_f^n$]
\label{prop-eig-bound-chro-cycle}
The chromatic number $\chi({C_{i}^{n}})$ is lower and upper bounded as 
 \begin{align}
 \label{eq:prop-up-and-lower-chromatic-cycles}
     1-\frac{2+\sum_{j=1}^{n-1}(2V^j)}{\max\Big(-\sqrt{\frac{V^n}{2}\cdot \frac{(V^n+1)}{2}} \ , -\sqrt{2{E}^n\cdot\frac{(V^{n}-1)}{2}} \Big)}\leq \chi({C_{i}^{n}})\leq  \sum_{j=1}^{n-1}(2V^j)\ +3 \ .
  \end{align} 
\end{prop}
\begin{proof}
\label{proof:bound-choromatic-cycles-max}
Combining (\ref{eq-chro_up-low_eig}), which bounds $\lambda_1$ and $\lambda_V$ of ${\bf A}_f$, with Proposition~\ref{prop-largest-eigenvalue}, as well as (\ref{hong-low-bound-small-eigen}) that lower bounds $\lambda_V$ \cite{hong1988bounds}, we have $1-\frac{2+\sum_{j=1}^{n-1}(2V^j)}{\lambda_{V^n}}\leq \chi({C_i^n})\leq \lfloor 2+\sum_{j=1}^{n-1}(2V^j)\rfloor+1$. We further simplify $\lfloor \cdot \rfloor$, because $\lambda_1$ is a positive integer~(see Proposition~\ref{prop-largest-eigenvalue}). Finally, substituting $\lambda_{V^n}$ with the maximum of the lower bounds in (\ref{birgham-low-bound-small-eigen}) and (\ref{hong-low-bound-small-eigen}) leads to (\ref{eq:prop-up-and-lower-chromatic-cycles}).
\end{proof}

To complement Proposition~\ref{prop-eig-bound-chro-cycle}, we next derive another bound that depicts the relation between the degree of each node in $C_i^n$ and $\chi({C_i^n})$.

\begin{cor}[Bounding $\chi({C_{i}^n})$ using degrees of $C_i^n$]
\label{cor_bound_choromatic_degree}
$\chi({C_{i}^n})$ satisfies the following relation:
\begin{align}
\label{eq-prop_bound_choromatic_degree}
 1+\frac{deg({x}^{n})}{\sqrt{2{E}^{n}-(V^n-1)\cdot(deg({x}^{n}))+(1+\sum_{i=1}^{n-1}2V^i)\cdot(deg({x}^{n}))}}
   \leq \chi({C_{i}^n})\leq \lfloor deg({x}^{n})\rfloor +1 \ .   
\end{align}
\end{cor}

\begin{proof}
\label{proof-prop_bound_choromatic_degree}
To characterize $deg({x}^n)$, we apply (\ref{eq-degree_function}) from Proposition~\ref{prop-deg-cycles}, and to bound $\chi({C_i^n})$, we use (\ref{eq-chro_up-low_eig}). Then, we apply the lower bound on $\lambda_{V^n}$ given in~\cite{das2004some}, which yields
\begin{align}
\label{eq-smallest-eigen-das-lowerbound}
\lambda_{V}({\bf A}_f)\ge -\sqrt{2{E}-(V-1)\cdot(\min_{{x}\in [V]} deg({x}))+\big(\min_{{x}\in [V]} deg({x})-1\big)\cdot(\max_{{x}\in [V]} deg({x}))} \ ,
\end{align}
giving the lower bound in (\ref{eq-prop_bound_choromatic_degree}). For the upper bound, we use (\ref{eq:lambda_1_C_i_n_power}), where $deg(x^n)=\lambda_1({\bf A}_f^n)$.
\end{proof}

Next, we consider $d$-regular graphs, and characterize the vertex degrees and chromatic numbers for the $n$-fold OR products of $d$-regular graphs.


\subsection{From Cycles to \texorpdfstring{$d$}{}-Regular Graphs}
\label{sec:sub-d-reg}

Building on our analysis for $C_i^n$ in Sections~\ref{sec:coloring_cyclic_graphs}-\ref{sec:Eig-Adj-Mat-Cycl-Graph-}, we now focus on $d$-regular source characteristic graphs. Next, we derive a closed-form expression for the degree of $G_{d, V}^n$.

\begin{prop}[Degrees of vertices in $G_{d,V}^{n}$]
\label{prop-deg-regular}
Given a $d$-regular graph $G_{d,V}=G(\mathcal{V}, \mathcal{E})$, the degree of each vertex in the $n$-fold OR product, denoted by $G_{d,V}^{n}$, for $n\geq 2$, is expressed as: 
\begin{align}
\label{eq-degree_function-regular}
 deg({x}^n)= d\cdot\frac{V^n-1}{V-1}\ , \quad \forall {x}^{n}\in [V^n] \ . 
\end{align}
\end{prop}
\begin{proof}
 \label{proof-deg-regular}   
The proof follows from employing Definition~\ref{d-regular graph, and Cycle graphs} for $G_{d, V}$. For details, see Appendix~\ref{App:prop-deg-regular}.
\end{proof}
Using (\ref{eq-degree_function-regular}), for the $n$-fold OR product of $G_{d,V}$  where $V$ is even, we next determine $\chi({G_{d, V}^n})$.

\begin{prop}[The chromatic number of $G_{d, V}^n$]
\label{prop-chrom-d-regular}
    The chromatic number of the $n$-fold OR product of 
    $G_{d,V}$ with an even number of vertices, i.e., $V=2k$, $k\in\mathbb{Z}^+$, is determined as:
\begin{align}
\label{eq-theo-chrom-d-regular}
    \chi({G_{d,V}^n})= d^n , \quad n\geq 1 \ .
\end{align} 
\end{prop}

\begin{proof}
See Appendix~\ref{app:proof-theo-chrom-d-regular}.
\end{proof}
For example, the $3$-regular graph $G_{3,6}$ (see Figure~\ref{fig:3.regularpower}) has $\chi(G_{3,6})=3$. For $n=2$, there are $6$ sub-graphs, namely $\{G_{3,6}^2(l)\}_{l\in[6]}$, where $\chi(G_{3,6}^2(l))=3$ for all $l\in[6]$. Let us choose a set of vertices in $G_{3,6}^2$ belonging to $\{G_{3,6}^2(5), G_{3,6}^2(6), G_{3,6}^2(1)\}$. We observe that this chosen subset is a complete graph, indicating that $\chi({G_{3,6}^2})\geq 9$. Reusing the same colors for the vertices of the remaining sub-graphs ($G_{3,6}^2(2), G_{3,6}^2(3), G_{3,6}^2(4)$), we deduce that $\chi({G_{3,6}^2})=3^2=9$.

 \begin{figure}[htbp]
 \centerline{\includegraphics
 [width=0.36\linewidth,  height= 3.5cm]{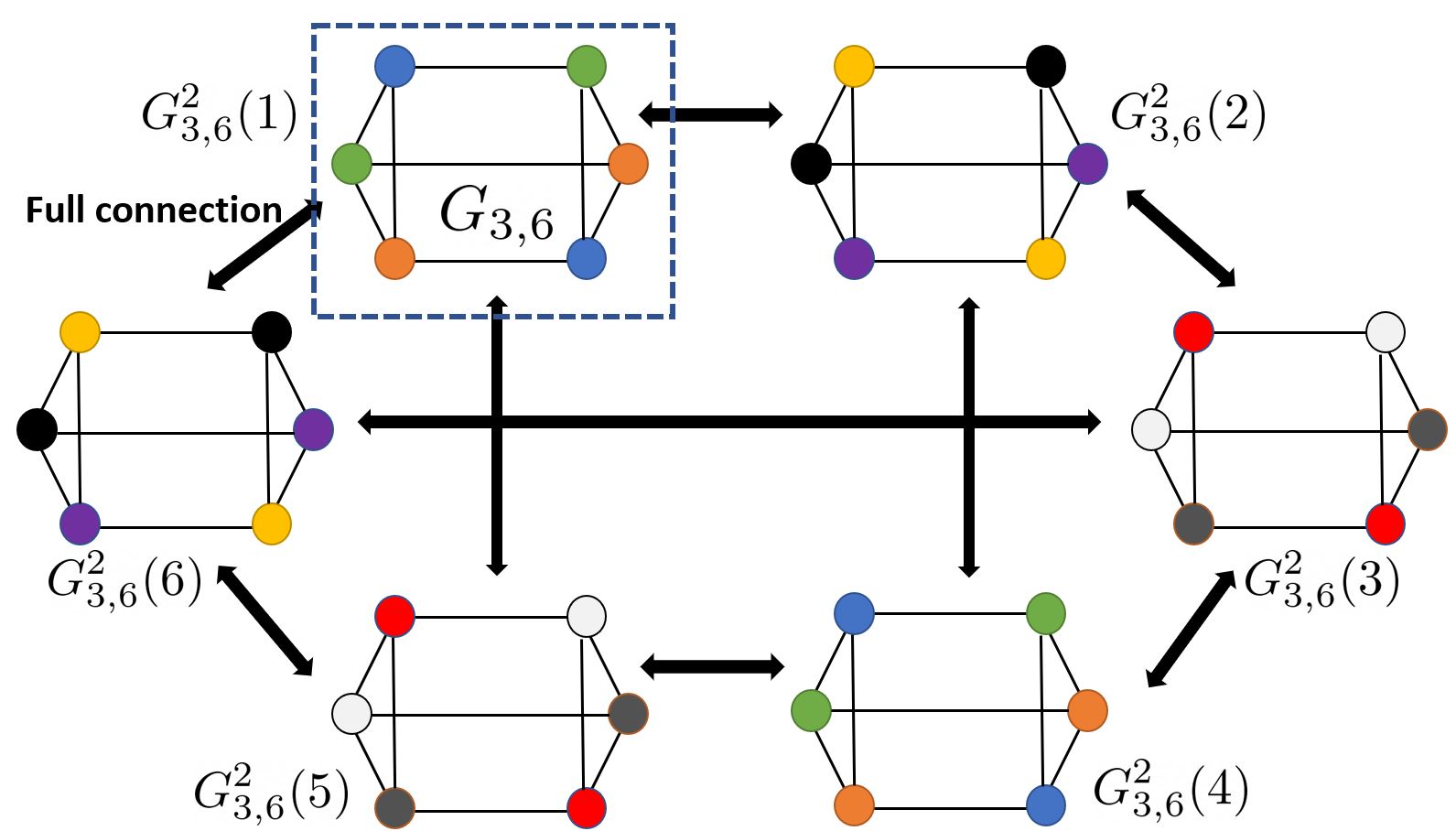}}
 \caption{A $3$-regular graph, $G_{3,6}$, is distinguished by a dashed square and $\chi({G_{3,6}^2})=9$.}
 \label{fig:3.regularpower}
 \end{figure}

Next, given $n$-fold OR products of $d$-regular graphs, $G_{d, V}^n$, we examine their expansion properties and present bounds on their \emph{expansion rates}.


\subsection{\texorpdfstring{$d$}{}-Regular Graphs and Graph Expansion}
\label{sec:graph-expansion}                          
Graph expansion quantifies how well-connected a graph is by measuring how easily subsets of vertices are connected to the remaining vertices of the graph. Graph expansion has applications in fields such as parallel computation \cite{ajtai1983sorting, arora1990line}, complexity theory, and cryptography \cite{bellare1993randomness, ajtai1987deterministic} due to its strong vertex connectivity and robustness properties
~\cite{kahale1992second}. In expander graphs, for any pair of distinct vertices $u, v\in [V]$, a path from $u$ to $v$ exists (see Definition~\ref{def-path}). We next examine expansion rates for several classes of expander graphs, including $G_{d, V}$, $C_i$, and $K_i$.

Given $G_{d, V}$, the second largest eigenvalue of its adjacency matrix ${\bf A}_f$, namely $\lambda_2({\bf A}_f)$, contributes to the linear expansion of $G_{d, V}$, where the number of edges grows linearly with the total number of vertices. 
Since the $n$-fold OR product of $G_{d, V}$ yields a $\deg(x^n)$-regular graph, it exhibits the connectivity properties of expander graphs. 


\begin{prop}[A lower bound on expansion rates of $G_{d, V}^n$]
\label{prop-expansion-G_dv}
The expansion rate of $G_{d, V}^n$, given its adjacency matrix ${\bf A}_f^n$, is lower bounded as follows:
\begin{align}
\label{eq:expansion}
\expan(G_{d,V}^n)\ge \frac{(d\cdot\frac{V^n-1}{V-1})^2}{\Lambda^2(G^n_{d,V})+\left((d\cdot\frac{V^n-1}{V-1})^2-\Lambda^2(G^n_{d,V})\right)\cdot{\frac{|Y|}{V^n}}}  \ ,
\end{align}
where 
\begin{align}
\label{lambda_Gn_dV}
\Lambda(G^n_{d,V})=\max \left(\lambda_2({\bf A}_f^n),\  |\lambda_{V^n}({\bf A}_f^n)|\right) \ ,
\end{align}
with $\lambda_{2}({\bf A}_f^n)$ and $\lambda_{V^n}({\bf A}_f^n)$ being the second and the smallest eigenvalues of ${\bf A}_f^n$, respectively.
\end{prop}

\begin{proof}
\label{pr: the expansion} 
As noted in Section~\ref{sec:model}, for a given $G_{d,V}$ with ${\bf A}_f$, $\lambda_1({\bf A}_f) = d$, and the eigenvalues of ${\bf A}_f^n$ for 
any connected $G_{d,V}^n $ 
are ordered as $\{d \geq \lambda_1 > \lambda_2 > \cdots > \lambda_{V^n}\}$. For any subset $Y$ of $G_{d,V}$, \cite{tanner1984explicit} shows that the size of $|N_G(Y)|$ satisfies the following: 
\begin{align}
\label{eq:exp 2}    
|N_G(Y)|\ge \frac{d^2\cdot|Y|}{\Lambda^2(G)+(d^2-\Lambda^2(G))\cdot{\frac{|Y|}{V}}} \ .
\end{align}

In (\ref{eq:exp 2}), $\Lambda(G)=\max (\lambda_2, |\lambda_V|)\leq d$, equality holds if and only if $G$ is bipartite, or is disconnected (composed of singleton vertices) \cite{kahale1992second}. 
Dividing both sides of (\ref{eq:exp 2}) by $|Y|$, letting $\expan(G^n)=|N_{G^n}(Y)|/{|Y|}$, and taking in the total number of vertices, $V^n$, in (\ref{eq:exp 2}) yields (\ref{eq:expansion}).  
\end{proof}

In Proposition~\ref{prop-expansion-G_dv}, given a graph $G_{d, V}$, $\Lambda(G_{d, V})$ captures its connectivity and structural balance. We infer that as $|Y|$ increases, the expansion rate is small, indicating less connectivity between vertices. Given a connected $G_{d, V}$, as evident from (\ref{eq-graph_expansion_formula}), the maximum expansion rate, denoted by $\expan^{ub}$, is achieved when $G_{d, V}^n$ is a complete graph, $K_i^n$, where $i\in[V^n]$, as follows:
\begin{align}
\label{eq-exapnd-complete-up-bound}
\expan^{ub}\geq\frac{({V^n-1})^2}{1+(({V^n-1})^2-1)\cdot{\frac{|Y|}{V^{n}}}} \ .
\end{align}

The minimum expansion rate, denoted by $\expan^{lb}$, is achieved when $G_{d, V}^n=C_i^n$ because $C_i^n$ has the minimum number of edges to ensure 
a Hamiltonian path for each $x\in[V]$. This yields
\begin{align}
\label{eq-cycle-expan-low-bound}
\expan^{lb}\geq\frac{(2\cdot\frac{V^n-1}{V-1})^2}{\lambda_{V^n}^2(C_i^n)+((2\cdot\frac{V^n-1}{V-1})^2-\lambda_{V^n}^2(C_i^n))\cdot{\frac{|Y|}{V^n}}} \ .
\end{align}

We next examine the relationship between $\lambda_1({\bf A}_f)$ and $\lambda_2({\bf A}_f)$ for sub-graphs of an expander graph, such as $G_{d,V}$. To this end, we apply~\cite[Lemma~3]{kahale1991better}, which states that for any subset of vertices, $Y \subseteq [V]$ in $G_{d, V}$, the induced sub-graph\footnote{An induced sub-graph is formed by selecting a subset of vertices $\subg \subseteq [V]$ includes all edges in $\mathcal{E}$ whose endpoints lie in $\subg$.} $\subg$ satisfies the following property:            
\begin{align}
 \label{eq:lam1-lam2}
 \lambda_1(\subg)\leq \lambda_2({\bf A}_f)+(d-\lambda_2({\bf A}_f))\cdot({|Y|}/{V}) \ .
\end{align}
For a connected $2$-regular graph ($C_i$), adapting (\ref{eq:lam1-lam2}) and selecting vertex sets $\subg$ and $Y$ such that $Y \cup \subg = C_i^n$, we establish a simplified relation between $\lambda_1({\bf A}_f^{n-1})$ and $\lambda_2({\bf A}_f^{n})$, as follows.
\begin{cor}[$\lambda_1({\bf A}_f^{n-1})$ versus $\lambda_2({\bf A}_f^{n})$ in $2$-regular graphs] 
\label{cor-relation-l1-l2-cycles}
The relationship between $\lambda_2({\bf A}_f^{n})$ of $C_i^n$ and $\lambda_1({\bf A}_f^{n-1})$ of $C_i^{n-1}$ is detailed as follows:
\begin{align}
\label{eq:remark lam1-2}  
\lambda_1({\bf A}_f^{n-1})\leq \lambda_2 ({\bf A}_f^{n})+\Big(2\cdot\frac{V^n-1}{V-1}-\lambda_2 ({\bf A}_f^{n})\Big)\cdot{\frac{1}{V}} \ .
\end{align}
\end{cor}

\begin{proof}
\label{proof-prop-relation-l1-l2-cycles} 
Given the relation between $\lambda_1(S)$ and $\lambda_2({\bf A}_f)$ for $G_{d,V}$ in (\ref{eq:lam1-lam2}), we consider that $G^n_{d,V}=C^n_i$, and we choose sub-graphs $Y$ and $S$ accordingly, as sub-graphs of $C_i^n$, where 
\begin{align}
\label{eq-proof-prop-relation-l1-l2-cycles} 
Y\subseteq \bigcup_{t\neq l ,\ t\in [V]} C_i^n(t)=\{C_{i}^{n}(1), C_{i}^{n}(2), \dots, C_{i}^{n}(V)\} \setminus \{C_i^n(l)\}  \ .
\end{align}
Using (\ref{eq:lam1-lam2}), by selecting $\subg\subset \{C_{i}^{n}(1), C_{i}^{n}(2), \dots,$ $C_{i}^{n}(V)\}$, namely $C_i^n(l)$, for the LHS of (\ref{eq:lam1-lam2}), in which each sub-graph is isomorphic to $C_i^{n-1}$, i.e., $\subg=\big(\bigcup_{t} C_i^n(t)\big)^{\mathsf{c}}$. Consequently, using (\ref{eq-proof-prop-relation-l1-l2-cycles}), $\subg$, and substituting them in (\ref{eq:lam1-lam2}), we reach the statement in (\ref{eq:remark lam1-2}).
\end{proof}

Next, we investigate general functions with arbitrary characteristic graphs. Building on Section~\ref{sec:results}, we present bounds on the chromatic number of their $n$-fold OR products.
 

\section{Bounds For General Characteristic Graphs}
\label{sec:sub-general-graph}

Here, we focus on general characteristic graphs $G(\mathcal{V},\mathcal{E})$, their $n$-fold OR products, $G^n$. Given a graph $G$ with an adjacency matrix ${\bf A}_f$, we first evaluate $deg(x_k)$ for $x_k\in [V]$, and derive bounds on $\lambda_1({\bf A}_f)$ and $\chi({G})$ in Section~\ref{sec:degrees_chrom_no_general_graphs}. We also present bounds on expansion rates and establish lower and upper bounds on entropies of $G^n$ in Sections~\ref{sec:expansion_general_graphs} and~\ref{sec:entropy_general_graphs}, respectively. Finally, for the $n$-fold OR product $G^n$, we introduce an approach that decomposes ${\bf A}_f^n$ to two symmetric block matrices and leverages GCT to investigate the spectrum of $G^n$ in Section~\ref{sec:gresh-bound}.

\subsection{Degrees and Chromatic Numbers of General Graphs}
\label{sec:degrees_chrom_no_general_graphs}

Given a general graph $G$, we next derive a recursive relation for $deg(x_k^n)$ for $x_k^n\in[V^n]$, where $deg(x^n_k)$ may vary across vertices, providing a generalization of Propositions~\ref{prop-deg-cycles}~and~\ref{prop-deg-regular}.

\begin{cor}[Degrees of vertices in $G^n$]
\label{cor-degree-or-power-general-graphs}
Given a general graph $G(\mathcal{V}, \mathcal{E})$, the degrees of vertices of $G^n$ are calculated as follows:
\begin{align}
\label{eq-theo-degree-or-power-general-graphs}
deg({x}_k^n)=deg({x}_k)+\sum_{j=1}^{n-1}deg({x}_k) \cdot V^j \ , \quad \forall x_k\in [V^n] \ . 
\end{align}
\end{cor}
\begin{proof}
Similarly to Propositions~\ref{prop-deg-cycles} and~\ref{prop-deg-regular}, we can compute $deg(x_k)$ for $x_k \in [V]$, with the distinction that each $x_k$ may have a different degree. For the $2$-fold OR product, each vertex, ${x}_k^2(l)$ for $l\in V$, connects to $deg(x_k)$ adjacent sub-graphs. Since neighboring nodes differ across $x_k$, we iteratively compute the degrees of $x_k^n$ separately for each $x_k$.
\end{proof}

Corollary~\ref{cor-degree-or-power-general-graphs} immediately implies that if vertices ${x}_t,\: {x}_k\in [V]$ for $t\neq k$ have equal degrees in $G^1$, i.e., $deg({x}_t)=\:deg({x}_k)$, then $deg({x}_t^n)=deg({x}_k^n)$ in the $n$-fold OR product, $G^n$, $n\geq 2$.  


Given a general graph $G$, we next devise lower and upper bounds on $\chi({G^n})$. To that end, we exploit the block matrix representation of ${\bf A}_f^n$ (see (\ref{eq-AdjacencyMatrix})) and use the maximum number of sub-matrices ${\bf J}_{V^{n-1}}$ in the rows of ${\bf A}_f^n$.

\begin{cor}[Bounds on $\chi({G^n})$ for a general $G$]
\label{cor-general_graph-approximation}
Given a general characteristic graph $G(\mathcal{V}, \mathcal{E})$, the chromatic number of $G^n$, $\chi({G^n})$, is bounded as follows:
\begin{align}
\label{eq-prop-general_graph-approximation}
1-\frac{\lambda_1({\bf A}_f)+d_{\max}\cdot\sum_{j=1}^{n-1}V^j} {\lambda_{V^n}({\bf A}_f^n)} \leq \chi({G^{n}}) \leq \left\lfloor {\lambda_1({\bf A}_f)+d_{\max}\cdot\sum_{j=1}^{n-1}V^j} \right\rfloor  \ ,
\end{align}
\end{cor}

\begin{proof}
The non-sparsity of an adjacency matrix, i.e., more $1$s in its entries, is directly related to its largest eigenvalue. The full connection between sub-graphs in the OR product is represented by ${\bf J}_V$, with $\lambda_1({\bf J}_V)=V$. Thus, the row with the most ${\bf J}_V$ matrices provides an upper bound for the largest eigenvalue. Additionally, the average number of ${\bf J}_{V^{n-1}}$ matrices across block rows of ${\bf A}_f^n$, multiplied by $V^{n-1}$ (which represents $\lambda_1({\bf J}_{V^{n-1}})$), approximates $\lambda_1({\bf A}_f^n)$. By modifying (\ref{eq-chro_up-low_eig}) and using $d_{\max}$ of $G$ to approximate $\lambda_1({\bf A}_f^n)$, we can establish a bound for $\chi({G^n})$.
\end{proof}

Corollary~\ref{cor-general_graph-approximation} refines the bounds in (\ref{eq-chro_up-low_eig}), Proposition~\ref{prop-eig-bound-chro-cycle} and   Corollary~\ref{cor_bound_choromatic_degree} by leveraging the exact value of $\lambda_1({\bf A}_f^n)$ and accounting for the specific structure of $G$. Given Corollary~\ref{cor-general_graph-approximation}, let us investigate the computational complexity for determining the eigenvalues $\lambda_k({\bf A}_f^n)$ and contrast it with the QR method\footnote{The QR transformation iteratively decomposes a matrix ${\bf A}_{(t)}$ where $t\in \mathbb{Z}^+$ denotes the iteration index, into an orthogonal matrix ${\bf Q}_{(t)}$ and an upper triangular matrix ${\bf R}_{(t)}$, satisfying ${\bf A}_{(t)} = {\bf Q}_{(t)} {\bf R}_{(t)}$. Then, the next iteration is given by ${\bf A}_{(t+1)} = {\bf R}_{(t)} {\bf Q}_{(t)}$. Under typical conditions (e.g., ${\bf A}$ is diagonalizable with distinct eigenvalues), it converges to an upper triangular matrix whose diagonal approximates $\lambda_k({\bf A})$. For an $m \times m$ matrix, QR requires a computational complexity of $O(m^3)$~\cite{watkins1982understanding}.} for calculating $\lambda_k({\bf A}_f^n)$, with a computation complexity of $O(V^{3n})$. However, using Corollary~\ref{cor-general_graph-approximation}, the overall complexity remains at $O(V^3)$. This is because eigenvalue computation for $G$ has a complexity of $O(V^3)$, the summations in RHS and LHS of (\ref{eq-prop-general_graph-approximation}) has $O(n)$, and the maximization step (determining $d_{\max}$) has $O(V)$. Hence, the dominant term, i.e., $O(V^3)$, dominates the final complexity.

Next, given a general $G(\mathcal{V},\mathcal{E})$ with an adjacency matrix ${\bf A}_f$, we derive lower and upper bounds on $\lambda_1({\bf A}_f^n)$ using Lemma~\ref{eigenvalue_dependency}, where the lower and upper bounds are functions of the minimum and maximum values of $deg(x_k)$ for $x_k\in[V]$. 

\begin{cor}[Bounds on $\lambda_1({\bf A}_f^n)$ for a general $G$]
\label{Cor-bound-chor-using-deg-max-min}
Given a general graph $G$, the largest eigenvalue for the adjacency matrix of $G^n$, denoted by $\lambda_1({\bf A}_f^n)$, is bounded as follows:
\begin{align}
 \label{eq-note-matrix-j-z}
 \min_{k \in [V]} (deg(x_k)) \cdot \lambda_1({\bf J}_{V^{n-1}}) \leq \lambda_1({\bf A}_f^n) \leq \max_{k \in [V]} (deg(x_k)) \cdot \lambda_1({\bf J}_{V^{n-1}}) \ ,
\end{align}
where we infer that
\begin{align}
\label{eq-Corr-matrix-j-z-avrage}
\lambda_1({\bf A}_f^n) \approx d_{\rm avg}(x_k) \cdot \lambda_1({\bf J}_{V^{n-1}}) \ . 
\end{align}
\end{cor}

\subsection{Bounds on Expansion Rates of General Graphs}
\label{sec:expansion_general_graphs}
Here, given a general graph $G$, we investigate the expansion of $G^n$. We exploit (\ref{eq-exapnd-complete-up-bound}), derived from the characteristics of the $n$-fold OR products of complete graphs, to obtain the upper bound, $\expan^{ub}$, and (\ref{eq-cycle-expan-low-bound}), derived from the $n$-fold OR products of cycles, to obtain the lower bound, i.e., $\expan^{lb}$. We next establish lower and upper bounds on $\expan(G^n)$.

\begin{cor}[Bounds on $\expan(G^n)$]
\label{cor-expansion_upper_bound_and _lower_bound}
The expansion rate for the $n$-fold OR product of a general characteristic graph $G(\mathcal{V}, \mathcal{E})$ is lower and upper bounded as follows:
\begin{align}
\label{eq_expansion_upper_bound_and _lower_bound} 
\expan^{lb}\leq \expan(G^n) \leq \expan^{ub} \ ,
\end{align}
where $\expan^{ub}$, derived from $K_i^n$ (representing a fully connected characteristic graph), 
and $\expan^{lb}$, from $C_i^n$ (representing a connected graph with the minimum number of edges), for $i=V^n$.
\end{cor}

\begin{proof}
See Appendix~\ref{app:proof_expansion_upper_bound_and_lower_bound}.
\end{proof}

Recall that the lower bound for $\expan(G_{d,V}^n)$ in (\ref{eq:expansion}) is given in terms of $\Lambda(G^n)=\max(\lambda_2({\bf A}_f^n),$ $ |\lambda_{V^n}({\bf A}_f^n)|)$. Whereas in Corollary~\ref{cor-expansion_upper_bound_and _lower_bound}, we use exact values of $\Lambda(K_i^n)$ and $\Lambda(C_i^n)$ for upper and lower bounds, respectively, with ${\bf A}_f^n$ denoting each graph’s adjacency matrix. Given $G^n$, $\expan (G^n)$ reflects its connectivity, with higher values leading to limited savings in source compression.

\subsection{Bounds on Entropies of General Graphs}
\label{sec:entropy_general_graphs}
Here, we derive upper and lower bounds on the graph entropies for general characteristic graphs. For the upper bound, we use a similar achievability approach as in the case of $C_i^n$ (see Proposition~\ref{prop-upper_entrop-cycle}), which relies on coloring the MISs of sub-graphs of $G^n$, i.e. $G^j$ for $j\in [n]$. For the lower bound, we employ fractional coloring applied to the $n$-fold OR products of general graphs to establish a bound on $H_G(X_1)$.  
We next derive an upper bound on $H_G(X_1)$.

\begin{prop}[An upper bound on $H_{G}({X}_1)$]
\label{prop_entropy_odd_cycles_MIS}
Given a characteristic graph $G(\mathcal{V},\mathcal{E})$, the entropy of $G^n$ is upper bounded as follows:
\begin{align}
\label{eq_Prop_MIS_based_entropy_general_graph}
H_{G}({X}_1) \leq \frac{1}{n} H\Big(\alpha_{n}\cdot \Big(\frac{|\text{MIS}_{G^n}|}{V^n}\Big), \alpha_{n-1}\cdot \Big(\frac{|\text{MIS}_{G^{(n-1})}|}{V^{n}}\Big), \cdots, \alpha_0\cdot\frac{1}{V^n}\Big) \ ,
\end{align}
where $\alpha_t \ , \ t\in [n]\cup \{0\}$ represents the number of MISs of $G^t$ with a cardinality of $|\text{MIS}_{G^{t}}|$.
\end{prop}

\begin{proof}
 See Appendix~\ref{App:Proof-of-Proposition_prop_entropy_odd_cycles_MIS}. 
\end{proof}

Despite the upper bound in Proposition~\ref{prop_entropy_odd_cycles_MIS}, to the best of our knowledge, with {\emph{traditional coloring schemes for $G(\mathcal{V},\mathcal{E})$}} where the total number of vertices is odd, i.e., $V = 2k+1$ for $k\in \mathbb{Z}^{\geq2}$, there is no established lower bound for $H_{G}({X}_1)$. To that end, in Corollary~\ref{cor-conv-graph-fraction}, we derive a lower bound on $H_{G}({X}_1)$ by employing the concept of {\emph{fractional coloring}} (see Definition~\ref{def-fractional-coloring}).

\begin{cor}[A lower bound on $H_{G}(X_1)$]
\label{cor-conv-graph-fraction}
The entropy of $H_{G}({X}_1)$ for a general connected characteristic graph $G(\mathcal{V},\mathcal{E})$ with $V=2k+1$ and Hamiltonian path, where $k\in\mathbb{Z}^{\geq2}$, and under uniform distribution of $X_1$, is lower bounded by 
    \begin{align}
    \label{eq-theo-conv-graph-fraction}
      \log_2 \left( \frac{2k+1}{k} \right)  \leq H_{G}({X}_1) \ .
    \end{align} 
\end{cor}

\begin{proof}
See Appendix~\ref{App:proof-theo-conv-graph-fraction}.
\end{proof}

From Corollary~\ref{cor-conv-graph-fraction}, we infer that employing fractional coloring, and (\ref{eq-n-fold-fractional}), yields a lower bound on $H_G(X_1)$. For $k=2$, the lower bound using $C_5$ for the graph $G(\mathcal{V}, \mathcal{E})$ with $V=5$ is given by $1.32\leq H_G(X_1)$, which matches the Shannon capacity of the pentagon~\cite{shannon1956zero}.



\subsection{Spectra of General Graphs} 
\label{sec:gresh-bound} 
Given a general graph $G$, we analyze the spectrum of ${\bf A}_f$ using the concept of GCT, as detailed in Definition~\ref{rem-gersh-circle}. We then exploit this spectrum to establish bounds on $\chi({G^n})$. 
For $G^n$, using the symmetry of ${\bf A}_f^n=({\bf A}_f^{n-1} \otimes {\bf I}_V + {\bf J}_{V}\otimes {\bf A}_f^{n-1} )\in \mathbb{F}_2^{V^n\times V^n}$, where $\otimes$ denotes the Kronecker product, we infer that circle $D_k$, in which an eigenvalue  $\lambda_k({\bf A}_f^n)$ is contained, simplifies to an interval:
\begin{align}
\label{eq-interval-gresh}
\delta_k=\{\lambda_k({\bf A}_f^n) \in \mathbb{R} : \left|\lambda_k ({\bf A}_f^n)-a^n_{kk}\right|\leq \sum_{t\neq k}  |a^n_{kt}| \}, \quad {k,t \in [V^n]} \ ,
\end{align}
$a^n_{kk}$ and $\sum_{t\neq k} |a^n_{kt}|$ denote the diagonal elements and the sums of non-diagonal elements in the $k$-th row of ${\bf A}_f^n$, respectively, then for $k =Vk'+i'$, and $t=Vt'+j'$, the elements are defined as $a^n_{kt} = (a^{n-1}_{k', t'} \cdot \Delta_{i'j'}+a_{i'j'}^{n-1}\mod 2)$ where $\Delta_{i'j'}$ denotes the Kronecker delta function (i.e., $\Delta_{i'j'}=1$ if $i'= j'$, and $\Delta_{i'j'}=0$ otherwise).

We next refine $\delta_k$ by exploiting the concept of block GCT, where using Definition~\ref{def-gresh-block} helps enclose each given $\lambda_k({\bf A}_f^n)$ within an interval denoted by $\delta_k^b$. Interval $k$ is characterized by the diagonal sub-matrices ${\bf A}_{kk}$, corresponding to ${\bf A}_f^{n-1}$, and the non-diagonal sub-matrices ${\bf A}_{kt}$, consisting of ${\bf Z}_{V^{n-1}}$ and ${\bf J}_{V^{n-1}}$, which represent disconnected and connected components of $G^n$, respectively. However, the block GCT representation leads to loose bounds on $\lambda_k({\bf A}_f^n)$ and subsequently on $\chi({G^n})$ via (\ref{eq-chro_up-low_eig}). To tighten these bounds, in Theorem~\ref{theo-eig-any-func} and Corollary~\ref{cor:iter-gresh-eigen}, we derive $\lambda_k({\bf A}_f^n)$ by leveraging GCT intervals from (\ref{eq-def-gersh-circle}) and (\ref{eq-block-gresh-disc-union}), and using a decomposition-based approach that splits ${\bf A}_f^n$ into two symmetric matrices.

\begin{theo}[$\lambda_k({\bf A}_f^n)$ via splitting ${\bf A}_f^n$ into two symmetric matrices]
\label{theo-eig-any-func}
The eigenvalues of ${\bf A}_f^n$, denoted by $\lambda_k({\bf A}_f^n)$, are given as follows:
   \begin{align}
    \label{eig-any-func-eq}
    \lambda_k({\bf A}_f^n) = \lambda_k ({\bf A}_{Gr}^n)+ \lambda_k ({\bf A}_{f^{\mathsf{c}}}^n) \ , \quad k\in [V^n],
   \end{align}
   where ${\bf A}^n_{Gr}$ is a block diagonal matrix with diagonal blocks formed from ${\bf A}_f^{n-1}$, and ${\bf A}_{f^{\mathsf{c}}}^n={\bf A}_f^n - {\bf A}^n_{Gr}$ captures the non-diagonal elements of ${\bf A}^n_{Gr}$.
\end{theo}

\begin{proof}
See Appendix~\ref{app:proof-eig-any-func}.
\end{proof}

Theorem~\ref{theo-eig-any-func} demonstrates that by decomposing ${\bf A}_f^n$ into ${\bf A}_{Gr}^n$ and ${\bf A}_{f^{\mathsf{c}}}^n$, we can capture the connections between ${\bf A}_f^{n-1}(l)$, $l\in[V]$, corresponding to the sub-graphs of $G^n$.

Next, we describe an iterative technique to determine $\lambda_k({\bf A}_f^n)$ from $\lambda_k({\bf A}_f)$. 

\begin{cor}
\label{cor:iter-gresh-eigen}
The eigenvalues of ${\bf A}_f^n$ can be iteratively calculated as follows:
\begin{align}
\label{cor-iter-gresh-eig-general-nth}
 \lambda_k ({\bf A}_f^n)= \lambda_k({\bf A}_f)+\sum_{j=2}^n \lambda_k ({\bf A}_{f^{\mathsf{c}}}^j) \:, \quad k \in [V^n] \ .
\end{align}
\end{cor}
\begin{proof}
See Appendix~\ref{App:cor-iter-gresh-eig-general-nth}. 
\end{proof}
Theorem~\ref{theo-eig-any-func} and Corollary~\ref{cor:iter-gresh-eigen} illustrate that leveraging the block structure of ${\bf A}_f^n$ reduces complexity compared to~\cite{echeverria2018block}. In Proposition~\ref{prop: chromatic-gresh}, we use Theorem~\ref{theo-eig-any-func} to derive a tighter bound on $\chi(G^n)$ than those from the block GCT representation intervals in (\ref{eq-block-gresh-disc-union}).

\begin{prop}[Bounds on $\chi({G^n})$ using GCT] 
\label{prop: chromatic-gresh} 
Given a general graph $G(\mathcal{V}, \mathcal{E})$, the chromatic number $\chi({G^n})$ is bounded as follows:
\begin{align}
\label{eq: prop:chromatic-gresh}    
1-\frac{\lambda_1 ({\bf A}_{Gr}^n)+ \lambda_1 ({\bf A}_{f^{\mathsf{c}}}^n)}{ -\sqrt{(V^{n}/2)\cdot[(V^{n}+1)/2]}}\leq \chi({G^n})
\leq \lfloor \lambda_1 ({\bf A}_{Gr}^n)+ \lambda_1 ({\bf A}_{f^{\mathsf{c}}}^n) \rfloor +1\ .
 \end{align}
\end{prop}

\begin{proof}
\label{proof-prop: chromatic-gresh}
To prove this result, we use the bounds for $\chi({G^n})$ in (\ref{eq-chro_up-low_eig}) using the eigenvalues of ${\bf A}_f^n$, adjust $\lambda_1({\bf A}_f^n)$  and $\lambda_{V^n}({\bf A}_f^n)$, and by utilizing (\ref{eig-any-func-eq}) from Theorem~\ref{theo-eig-any-func}, and (\ref{hong-low-bound-small-eigen}), respectively. 
\end{proof}

For the RHS and LHS in (\ref{eq: prop:chromatic-gresh}), (\ref{eq-Corr-matrix-j-z-avrage}) provides a tighter approximation on $\lambda_1({\bf A}_f^n)$ compared to (\ref{eq-note-matrix-j-z}), which stems from employing~\cite[Lemma~5]{mehlhorngreat}, and ${d_{\rm avg}}\cdot V^{n-1}\leq \left\vert\delta^b(\lambda_1({\bf A}_f^n))/{2}\right\vert$.


\begin{exm}
\label{exm_geresh}   
Consider a characteristic graph $G_1$ with an adjacency matrix 
\begin{align}
\label{matrix-example-gresh}
    {\bf A}_{f_1}=\begin{bmatrix}
        0 & 1 & 0 & 0 & 1\\
        1 & 0 & 1 & 1 & 0\\
        0 & 1 & 0 & 1 & 0\\
        0 & 1 & 1& 0 & 1\\
        1& 0 & 0 & 1 & 0 
    \end{bmatrix} \ ,
\end{align}
with a set of distinct eigenvalues  $\vartheta(G_1) = \{2.4812, 0.6889, 0, -1.1701, -2\}$. Using GCT, we derive five intervals $\{\delta_k\}_{k\in[5]}=\{[-2,2],\ [-2,2],\ [-2,2],\ [-3,3], [-3,3]\}$ for ${\bf A}_{f_1}$, one for each eigenvalue, where each $\delta_k$ is centered at $0$ since $trace({\bf A}_{f_1}^2) = 0$. Two unique intervals with the largest lengths are $\delta_1=[-2,2]$ and $\delta_2 =[-3, 3]$, which are used to determine $\lambda_1({\bf A}_{f_1})$.

From Theorem~\ref{theo-eig-any-func} and Corollary~\ref{cor:iter-gresh-eigen}, we have $\lambda_1({\bf A}_{f_1}^2) \in [12, 18]$. Applying the GCT for block matrices (see (\ref{eq-block-gresh-disc-union}) in Section~\ref{prelem}), we obtain $\Xi({\bf A}_{f_1}^2)\in\cup_{k=1}^5 \delta^b_{k}$, where $\delta^b = [-18, 18]$ includes all possible eigenvalues but provides a less precise estimate than (\ref{eig-any-func-eq}). Refining the bounds for $\lambda_1({\bf A}_{f_1}^2)$ using the ${d}_{\rm avg}(x_k)$ in (\ref{eq-note-matrix-j-z}) and the upper bound in (\ref{eq: prop:chromatic-gresh}) gives a more precise interval of $[12, 15]$. This range shows the upper bound is $3$ units tighter than the maximum degree method.
\begin{figure}[h]
\centering
\includegraphics[width=0.65\linewidth,  height= 4.5cm]{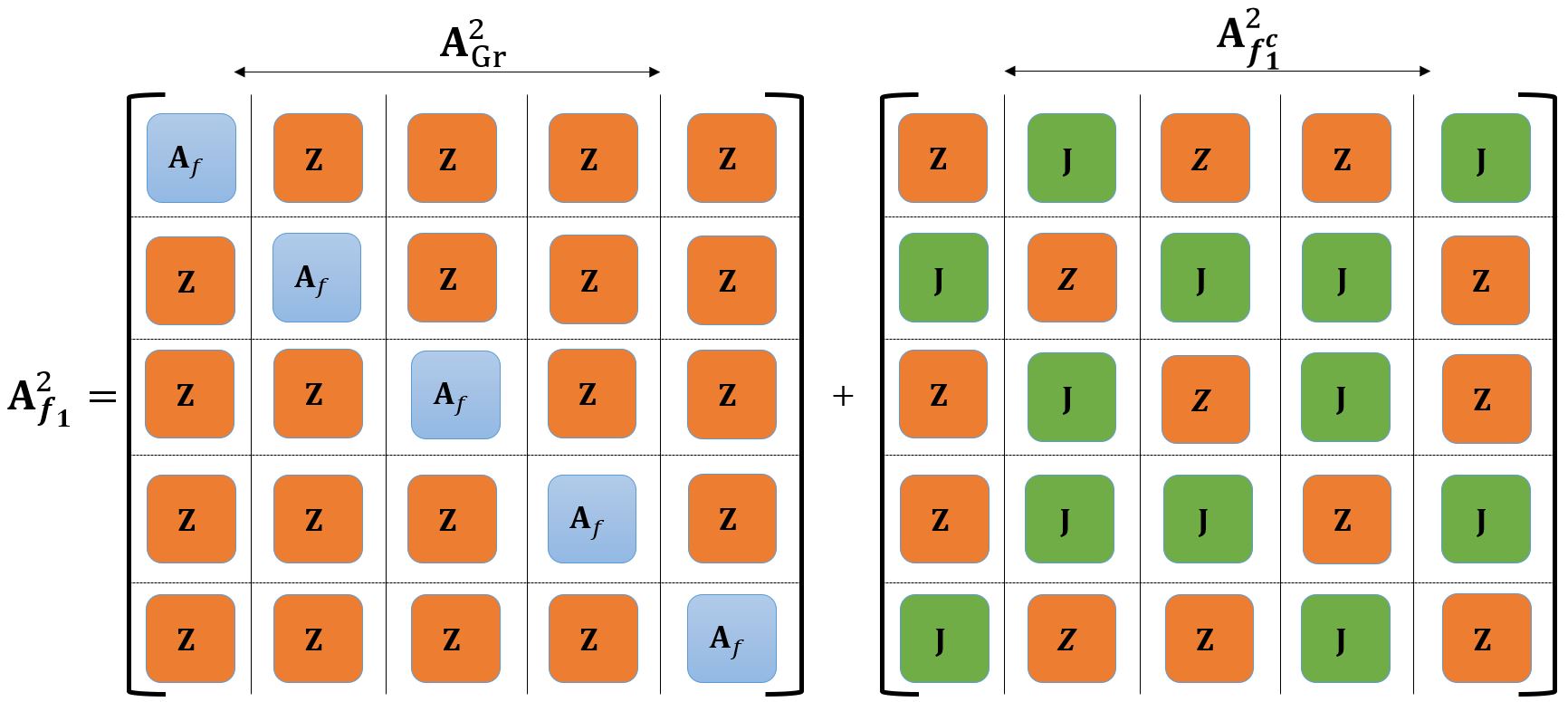}
\caption{Splitting the adjacency matrix ${\bf A}_{f_1}^2$ into two symmetric matrices ${\bf A}_{{Gr}}^2$ and ${\bf A}^2_{f_1^c}$.}
\label{fig:split-gresh}
\end{figure}
\end{exm}


\section{Conclusion }
\label{sec:conclusion}

In this paper, we addressed the problem of distributed functional compression by introducing novel coloring-based encoding schemes for source characteristic graphs. We analyzed various graph topologies—cycles ($C_i$), $d$-regular graphs ($G_{d, V}$), and general graphs ($G$)—and their $n$-fold OR product realizations ($C_i^n$, $G_{d, V}^n$, and $G^n$), exploring the interplay between adjacency matrix eigenvalues and chromatic numbers to develop low-entropy colorings schemes and derive bounds on the compression rate for asymptotically lossless function compression.

For cycles, we derived bounds on the degrees of $C_i^n$ and proposed a recursive coloring scheme for $C^n_{2k+1}$, which computes valid colorings in polynomial time by leveraging their structural properties. We also investigated the relationship between the spectra of $C_i^n$ and the chromatic numbers $\chi({C_i^n})$ to establish bounds on the minimum entropy colorings. 

For $d$-regular graphs, we analyzed the degrees and eigenvalues for $G_{d, V}^n$. We also investigated the connection between the OR products of $G_{d, V}$ and graph expansion described by the spectral properties of $G_{d, V}$. This enabled us to derive upper and lower bounds on the expansion rates of general characteristic graphs, where the expansion rate reflects the connectivity of sub-graphs, where higher connectivity implies an increase in $H^{\chi}_{G^n}({\bf X}_1)$.

For general characteristic graph topologies $G$, we derived bounds on the degrees and eigenvalues of $G^n$ using the block matrix representation of the adjacency matrix of $G^n$, denoted by ${\bf A}_f^n$, in order to achieve a reduced computational complexity for deriving the eigenvalues $\lambda_k({\bf A}_f^n)$ for $k\in [V^n]$ compared to the QR method~\cite{francis1961qr, watkins1982understanding}. By leveraging the GCT approach, we provided upper and lower bounds on $\lambda({\bf A}_f^n)$ and compared these with the iterative ${\bf A}_f^n$ decomposition method (see Theorem~\ref{theo-eig-any-func} and Corollary~\ref{cor:iter-gresh-eigen}), which exploits the properties of eigenvalues of symmetric matrices to produce tighter bounds for $\lambda({\bf A}_f^n)$.

Future directions include examining connections between functional compression and other graph properties, such as diameter, graph decomposition, and graph complement, and devising lossy compression schemes. Our compression approach captures various characteristic graph topologies, which can be exploited to represent unions of source characteristic graphs and user demands in multicast and broadcast computation settings to establish fundamental lower bounds on communication costs. 

\section*{Acknowledgment} 
The authors acknowledge the constructive comments of Dr. Kizhakke and Mr. Tanha. 


\appendix

\subsection{Proof of Proposition~\ref{prop-deg-cycles}}
\label{App:theo-deg-even-cycles}

Consider $C_{i}=G(\mathcal{V}, \mathcal{E})$, where for each $x\in[V]$, $deg(x)=2$ (see Definition~\ref{d-regular graph, and Cycle graphs}). The $2$-fold OR product is denoted by $C_{i}^{2}$, and it has $V$ sub-graphs $\{C^2_{i}(1),C_{i}^2(2),\dots,C_{i}^2(V)\}$. 
For a given sub-graph, say $C_i^2(l)$, $l\in[V]$, any ${x}^2(l)\in C_{i}^2(l)$, where ${x}^2(l)$ denotes a vertex in the $l$-th replica of $C^2_i$, is connected to any vertex of the adjacent sub-graphs, namely $\in \{C_{i}^2((l-1)\mod(V)), \: C^2_{i}((l+1)\mod(V))\}$. Therefore, each vertex in $C^2_i(l)$ has $V$ edges to $C^2_{i}((l-1)\mod(V))$ and $C^2_{i}((l+1) \mod(V))$ each. Accounting for the set of edges between adjacent sub-graphs, and each vertex's degree in the sub-graph itself, the degree of each vertex in $C_i^2$ is $deg({x}^2)=2+2V$. For $n=3$, this results in $deg({x}^3)=2+2V+2{V}^2$.

For a given $C_{i}^{n-1}(l)$, each vertex ${x}^{n-1}(l)$ is connected to all vertices in the adjacent sub-graphs, $C_{i}^{n-1}((l-1) \mod(V))$ and $C_{i}^{n-1}((l+1)\mod(V))$. Similarly, the degree for the $(n-1)$-fold OR product, by considering the previous OR product's degree and using induction, is going to be $deg({x}^{n-1}) = 2+2{V}+2{V}^2+\dots+2{V}^{n-2}=2+\sum_{j=1}^{n-2}2 ({V}^j)$. Therefore, for building $C_i^n$, using the same procedure, there are ${V}$ sub-graphs $\{C_{i}^{n}(1), C_{i}^{n}(2)$ $,\dots, C_{i}^{n}({V})\}$. Like $C_i^2$, the sub-graph $C_{i}^{n}(l)$, $l\in[{V}]$, is fully connected to the adjacent sub-graphs $\{C_{i}^{n}((l-1)\mod(V)),\: C_{i}^{n}((l+1)\mod(V))\}$. Therefore, for $C_i^n$, $deg({x}^{n})$ is calculated using (\ref{eq-degree_function}).




\subsection{Proof of Proposition \ref{prop-chromatic_number_cycles-odd}}
\label{App:chromatic_number_cycles}
Given a cycle $C_5$ with $\chi(C_5)=3$, for coloring $C_{5}^2$, $\chi(C_{5}^2)=8$, as shown in Figure~\ref{fig:2nd_power_C5}, saving one color compared to the greedy algorithm's $(\chi(C_5))^2=9$ colors. Instead of coloring the entire graph, our approach assigns coloring sets to sub-graphs, for finding an optimal coloring. The coloring sets for valid coloring of sub-graphs, namely $\mathcal{C}_i(l)$ for sub-graph $l$, ensure that no two neighboring sets share colors. To achieve valid colorings, vertices are divided into sub-graphs, each assigned a set of colors. For $C_5$, a valid coloring requires three distinct colors $\{c_1,c_2,c_3\}$.

The cardinality of the coloring set $\mathcal{C}^{j}_i(l)$ changes based on the chromatic number of the previous OR product $\chi({C_i^{j-1}})$. 
However, the number of coloring sets is always constant and equal to the number of vertices in $C_i$, for $i=V$. Given $C_5$, $5$ coloring sets are assigned to the sub-graphs. To cover the sub-graphs assigned to the first two color sets, the adjacent coloring sets cannot share the same colors.
Thus, $\{c_1,c_2,\dots,c_6\}$ are required in the first two coloring sets, namely $\mathcal{C}_{5}^{2}(1)$ and $\mathcal{C}_{5}^{2}(2)$.
Consequently, we can reuse the colors from the first set, since there is no edge between sub-graphs $C_{5}^{2}(1)$ and $C_{5}^{2}(3)$. Hence, by adding only $2$ colors to the third color set and using them in a cyclic manner, all vertices are colored with only $8$ colors, allowing us to express $\mathcal{C}_5^2(l)$ for $l\in[5]$ as: $\mathcal{C}_{5}^{2}(1)=\{{c_1,c_2,c_3}\}, \: \mathcal{C}_{5}^{2}(2)=\{{c_4,c_5,c_6}\}, \: \mathcal{C}_{5}^{2}(3)=\{ {c_7,c_8,c_1}\}$, $ \mathcal{C}_{5}^{2}(4)=\{ {c_2,c_3,c_4}\}, \: 
\mathcal{C}_{5}^{2}(5)=\{{c_5,c_6,c_7}\}$.    
For $C_5^3$, we have $|\mathcal{C}_5^3(l)|= 8$, and applying the same method yields $\chi(C_5^3) = 20$. Thus, $\mathcal{C}_5^3(l)$ for $l \in [5]$ can be expressed as:
\begin{align*}
 \mathcal{C}_{5}^{3}(1)&=\{{c_1,c_2,c_3,c_4,c_5,c_6,c_7,c_8}\}, \quad      
 \mathcal{C}_{5}^{3}(2)=\{{c_{9},c_{10},c_{11},c_{12},c_{13},c_{14},c_{15},c_{16}}\},              \\
\mathcal{C}_{5}^{3}(3)&=\{{c_{17},c_{18},c_{19},c_{20},c_1,c_2,c_3,c_4}\},     \quad
 \mathcal{C}_{5}^{3}(4)=\{{c_5,c_6,c_7,c_8,c_9,c_{10},c_{11},c_{12}}\}, \\
\mathcal{C}_{5}^{3}(5)&=\{{c_{13}},c_{14},c_{15},c_{16},c_{17},c_{18},c_{19},c_{20}\}.     
\end{align*} 
Using induction we show that $\chi({C_{5}^{1}})=3$, $\chi({C_{5}^{2}})=8$, $\chi({C_{5}^{3}})=20$, $\chi({C_{5}^{4}})=50$, $\chi({C_{5}^{5}})=125$,  $\chi({C_{5}^{6}})=313$, and so on. Following this pattern for $n+1$, we reach $\chi({C_{i}^{n+1}})=2\chi({C_{i}^{n}})  +\left\lceil{\frac{\chi({C_{i}^{n}})}{2}}\right\rceil$.


\subsection{Proof of Proposition~\ref{prop_gain_over_greedy}}
\label{App:prop_gain_over_greedy}
From (\ref{eq-chromatic_number_cycles}) in Proposition~\ref{prop-chromatic_number_cycles-odd}, for $C_{2k+1}^n$ with $k\geq 2$, we infer that
\begin{align}
\label{eq-choromatic-cyclic-lower-upper-bound}
2\chi({C_{2k+1}^{n-1}})+\frac{\chi({C_{2k+1}^{n-1}})}{2}\leq\chi({C_{2k+1}^n}) \leq 2\chi({C_{2k+1}^{n-1}})+\frac{\chi({C_{2k+1}^{n-1}})}{2}+1 \: .\
\end{align}

Exploiting the number of colors needed for coloring $C_{2k+1}^n$ using greedy algorithm, which is given as $\big(\chi({C_{2k+1}})\big)^n=3^n$, and (\ref{eq-choromatic-cyclic-lower-upper-bound}), the gain, $\eta_n$, lies in the following interval:
\begin{align}
\label{eq-gain-chormatic-recuersive}
    \frac{3^n}{2\chi({C_{2k+1}^{n-1})}+\frac{\chi({C_{2k+1}^{n-1}})}{2}+1}\leq\eta_n\leq \frac{3^n}{2\chi({C_{2k+1}^{n-1}})+\frac{\chi({C_{2k+1}^{n-1}})}{2}}  \:\:.\ 
\end{align}

For $n=1$, we have $\chi({C_{2k+1}}) =\chi_{\text{Greedy}}({C_{2k+1}})=3=\chi_1$. 
For $n\geq 2$, $\chi_{\text{Greedy}}({C_{2k+1}^n})=3^{n-1} \chi_1 $, whereas our cyclic approach yields
\begin{align}
\label{eq-UB-chromatic-ours}
\chi({C_{2k+1}^n})\geq\Big(\frac{5}{2}\Big)^{n-1}\chi_1 \ , \quad n\geq 2\: .\    
\end{align}

Substituting $\chi({C_{2k+1}^{n-1}})\geq \left(\frac{5}{2}\right)^{n-2}\chi_1$ in the denominator of the upper bound in (\ref{eq-gain-chormatic-recuersive}), where the ratio of the nominator to the denominator leads to $\eta_n \leq \frac{3^{n-1}}{(\frac{5}{2})^{n-1}}$. Furthermore, using (\ref{eq-UB-chromatic-ours}), the denominator of the lower bound in (\ref{eq-gain-chormatic-recuersive}) is lower bounded as: 
\begin{align}
\label{Eq-LB-proofchromatic}
    2\chi({C_{2k+1}^{n-1}})+\frac{\chi({C_{2k+1}^{n-1}})}{2}+1 &\geq 2\Big(\frac{5}{2}\Big)^{n-2}\chi_1+ \frac{1}{2} \Big(\frac{5}{2}\Big)^{n-2}\chi_1+1=\Big(\frac{5}{2}\Big)^{n-1}\chi_1+1 \ ,
\end{align}
and from (\ref{Eq-LB-proofchromatic}),
$\eta_n\geq \frac{3^{n-1} \chi_1}{\left(\frac{5}{2}\right)^{n-1} \chi_1 + 1}$, and as $n \to \infty$, the lower and upper bounds match.   

\subsection{Proof of Proposition~\ref{prop-upper_entrop-cycle}}
\label{App:corr_Entropy_bound_odd}

By using the recursive coloring approach for $\chi({C_{2k+1}^n})$ in (\ref{eq-chromatic_number_cycles}) and calculating the coloring PMFs assigned to MISs, we derive the upper bound for $H_{C_{2k+1}}({X}_1)$. The cardinality of MISs for $C_{2k+1}$ is $|\text{MIS}_{C_{2k+1}}|=k$. Given the full connection of sub-graphs in the $n$-fold OR product, we deduce that in $C_{2k+1}^n$, $|\text{MIS}_{C_{2k+1}^n}|=k^n$. To derive the upper bound, we apply constraints on $\alpha_t$, representing the coefficients for coloring MISs and calculating $\min_{\mathcal{C}_{C_i^n}} H(\mathcal{C}_{C_i^n})$. The chromatic entropy is bounded by determining $|\text{MIS}_{C_{2k+1}^n}|$ for each sub-graph in $C_{2k+1}^j(l)$, where $j \in [n]$ and $l \in [V]$, while satisfying the valid coloring constraints described as follows.
\begin{align}
\label{eq-proof-prop-ent-constraints-cycles}
    (a)&\quad \Big(\sum_{t=1}^{n} \alpha_t\cdot\frac{k^t}{V^n}\Big)+\alpha_0\cdot\frac{1}{V^n}=1,  \:\:   
 (b)\quad \alpha_t\geq k\cdot \alpha_{t-1} \ , \quad (c)\quad  \alpha_{0}=1, 
\end{align}
$(a)$ illustrates the PMF 
of colors assigned to MISs. In each sub-graph, $|\text{MIS}_{C_{2k+1}^j}|$ represents the maximum independent set size of $C_{2k+1}^j(l)$ for $l \in [V]$. $(b)$ ensures larger MISs are assigned more colors to minimize the total number of colors used. $(c)$ indicates $|\text{MIS}_{C_{2k+1}^0}|$ refers to the MIS of a singleton, representing a color assigned to a single node. Since, when $V=2k+1$ is odd, the sum of cardinalities of independent sets in $C_{2k+1}^j$ requires an odd object. Then, by incorporating constraints described in (\ref{eq-proof-prop-ent-constraints-cycles})-$(a)$, and $(b)$, we have 
\begin{align}
\label{eq_proof_upper_bound_polynomial}
\alpha_{n} \cdot \Big(\frac{k^n}{(2k+1)^n}\Big) + \alpha_{n-1}\cdot \Big(\frac{k^{(n-1)}}{(2k+1)^{n}}\Big)+ \cdots + \alpha_0 \cdot\Big( \frac{1}{(2k+1)^n}\Big) \ge 1 \ .   
\end{align}

Using (\ref{eq_proof_upper_bound_polynomial}), and properties of geometric series, we obtain
\begin{align}
\label{eq-proof-cor-up}    
1=\sum_{t=0}^n \frac{\alpha_t k^t}{(2k+1)^n}< \frac{\alpha_n}{(2k+1)^n}\cdot \sum_{t=0}^{n} \frac{k^t}{k^{n-t}} = \frac{\alpha_n}{(2k+1)^n}\cdot\frac{1}{k^n} \cdot \frac{(k^2)^{(n+1)}-1}{k^2-1} \ ,
\end{align} 
where the final equality follows from computing the summation, 
which leads to $\alpha_n> \frac{(2k+1)^n k^n(k^2-1)}{k^{2(n+1)}-1}$. Next, to derive the upper bound on $\alpha_n$, we assign $\alpha_1=0$ and maximize $\alpha_n$. From (\ref{eq-proof-cor-up}), we have 
\begin{align}
\label{eq-proof-cor-up_up_alpha_n}    
1=\sum_{t=0}^n \frac{\alpha_t k^t}{(2k+1)^n}>\frac{\alpha_n}{(2k+1)^n} \cdot \sum_{t=0}^{n} \frac{k^t}{k^{n-t}}  \ , 
\end{align}
where reordering (\ref{eq-proof-cor-up_up_alpha_n}), we have the following  upper bound on $\alpha_n$:
\begin{align}
\label{proof_upper_bound_alpha_up}
 \alpha_n< \frac{(k^{2(n+1)}-1)\cdot(2k+1)^n \cdot k^{(n-1)}}{(2k+1)^{2n}\cdot(k^2-1)\cdot k^{(n-1)}-(k^{2(n+1)}-1)} \ .
\end{align}

Similarly, using (\ref{proof_upper_bound_alpha_up}) we can bound $\alpha_t$ where $t< n$, which concludes our proof.



 \subsection{Proof of Theorem \ref{cyclic_adjacency}}
\label{App:cyclic_adjacency}
Assume ${\bf A}_f$ has $V$ eigenvalues. To find the eigenvalues of ${\bf A}_f^2$, one must solve (\ref{eq-row-cycle-matrix-id}), derived from evaluating (\ref{Af2_eigenvalues}), yielding $V$ equations per block row and $V^2$ equations in total. Hence, there are $V^2$ equations, from which the remaining $(V-1)\times V$ equations are just replicas of the eigenvalues of any given block row, due to cyclic symmetry. In (\ref{eq-row-cycle-matrix-id}), the terms with indices $k+1$ and $k-1$ are in modulo $V$. Two additional equations are necessary to calculate $\lambda_k({\bf A}_f^j)$ where $k\in [V^j]$, as the power $j$ increases from $n-1$ to $n$. Both ${\bf A}_f$ and ${\bf J}_V$ are diagonalizable, i.e., there exists an invertible matrix ${\bf P}$ such that ${\bf A}_f={\bf P}^{-1}{\bf H}_{{\bf A}_f}{\bf P}$ and ${\bf J}_V={\bf P}^{-1}{\bf H}_{{\bf J}_V}{\bf P}$, where ${\bf H}_{{\bf A}_f}$ and ${\bf H}_{{\bf J}_V}$ are diagonal matrices. Hence, the sum ${\bf A}_f+{\bf J}_V$ satisfies 
\begin{align}
{\bf A}_f+{\bf J}_V={\bf P}^{-1}{\bf H}_{{\bf A}_f}{\bf P}+{\bf P}^{-1}{\bf H}_{{\bf J}_V}{\bf P}={\bf P}^{-1}({\bf H}_{{\bf A}_f}+{\bf H}_{{\bf J}_V}){\bf P} \ .\nonumber
\end{align}

The eigenvalues of ${\bf A}_f+{\bf J}_V$ can now be computed using the diagonal of ${\bf H}_{{\bf A}_f}+{\bf H}_{{\bf J}_V}$, where $\lambda_k({{\bf A}_f})$ and $\lambda_k({\bf J}_V)$ represent the eigenvalue of ${\bf A}_f$ and ${\bf J}_V$ matrices for $k\in [V]$, respectively.
\begin{align}
\label{eq:eigen_larg}
\lambda_k ({\bf A}_{f}+{\bf J}_V)=\lambda_k({{\bf A}_f})+\lambda_k({\bf J}_V), \quad k\in [V]  \ .
\end{align}

From Lemma \ref{eigenvalue_dependency}, the number of distinct eigenvalues of ${\bf A}_{f}^{2}$ differs from the eigenvalues of ${\bf A}_f$ at most by two, and similarly for the eigenvalues of ${\bf A}_f^j$, $j\in \mathbb{Z}^{+\ge 2}$ derived from ${\bf A}_f^{j-1}$ of the $(j-1)$-fold OR product graph. 
From (\ref{eq-AdjacencyMatrix}), the sub-matrices ${\bf A}_{f}^{j-1}$ in ${\bf A}_f^j$ take the following form 
\begin{align}
\label{adjacency_cyclic_relation}
{\bf A}_f^{j-1} {\bf v}_k+ {\bf J}_{V^{j-1}} {\bf v}_{k+1}+ {\bf J}_{V^{j-1}} {\bf v}_{k-1}=\nu {\bf v}_k,\quad k\in[V]\ , 
\end{align}
where $j \in {\mathbb{Z}^{+\ge2}}$, and the column vector ${\bf v}_k$ has dimensions $V^{j-1} \times 1$. 
For the $n$-fold OR product, ${\bf A}_f^n$, is constructed by combining ${\bf A}_f^{n-1}$, ${\bf J}_{V^{n-1}}$ from the $(n-1)$-fold OR product, and ${\bf Z}_{V^{n-1}}$ matrices, as in (\ref{eq-AdjacencyMatrix}). 
The first block row of ${\bf A}_f^n$ includes ${\bf A}_f^{n-1}$, two ${\bf J}_{V^{n-1}}$ matrices representing full connections to adjacent sub-graphs, and ${\bf Z}_{V^{n-1}}$ matrices indicating no connections, all of which affect eigenvalue computation. Thus, $C_i^n$ has two more distinct eigenvalues than $C_i^{n-1}$. 

   




\subsection{Proof of Proposition~\ref{prop-deg-regular}}
\label{App:prop-deg-regular}

Given $G_{d,V}=G(\mathcal{V},\mathcal{E})$, its $2$-fold OR product is denoted by $G_{d,V}^2$, consisting of $V$ sub-graphs, $G_{d,V}^2=\{G^2_{d,V}(1),$ $G^2_{d,V}(2),$ $\dots, G_{d,V}^2(V)\}$. For a given sub-graph $G^2_{d,V}(l)$, where $l\in[V]$, any vertex ${x}^2(l)$ is connected to vertices in $d$ adjacent sub-graphs between following indices $\left( l-\lceil\frac{d}{2}\rceil,  l+\lceil\frac{d}{2}\rceil\right)$, i.e., $\in\{G^2_{d,V}((l-\lceil\frac{d}{2}\rceil)\mod(V)), \dots,$ $G^2_{d,V}((l-1)\mod(V)), G^2_{d,V}((l+1) \mod(V)), \dots, G^2_{d,V}((l+\lceil\frac{d}{2}\rceil) \mod(V))\}$. Thus, each vertex in $G^2_{d, V}(l)$ has $d \times V$ edges to these adjacent sub-graphs. For $n=2$, $deg({x}^2)=d +dV$, accounting for both edges within sub-graphs and edges to adjacent sub-graphs. For $n=3$, accounting for edges to adjacent sub-graphs and those within each sub-graph (i.e., $deg({x}^2)$) yields $deg({x}^3)= d+dV+dV^2$. 
For a given $G_{d,V}^{n-1}(l)$, each vertex ${x}^{n-1}(l)$ is connected to all vertices in the adjacent sub-graphs $\{G^{n-1}_{d,V}((l-\lceil\frac{d}{2}\rceil)\mod(V)), \dots,G^{n-1}_{d,V}((l-1)\mod(V)), G^{n-1}_{d,V}((l+1)\mod(V)), \dots$ $, G^{n-1}_{d,V}((l+\lceil\frac{d}{2}\rceil)\mod(V))\}$. Similarly, for the $(n-1)$-fold OR product, $deg({x}^{n-1})$ is calculated by induction, starting from the degree in the previous product, which yields $deg({x}^{n-1})= d+dV+d{V}^2+\dots+d{V}^{n-2}=d\cdot\sum_{j=0}^{n-2} ({V}^j)$. Therefore, for $G_{d, V}^{n}$, there are ${V}$ sub-graphs, where $G_{d,V}^{n}(l)$, $l\in[{V}]$, is fully connected to the adjacent sub-graphs, and for ${x}^{n}\in [V^n]$, $deg({x}^n)= d+d\cdot \sum_{j=1}^{n-1} V^j$.
Thus, $deg({x}^n)$ is given in closed form by (\ref{eq-degree_function-regular}).

\subsection{Proof of Proposition~\ref{prop-chrom-d-regular}}
\label{app:proof-theo-chrom-d-regular}
 For any graph $G(\mathcal{V}, \mathcal{E})$ where $\max_{k \in [V]} deg(x_k) \geq 3$ and no complete sub-graph exists of size $\max_{k \in [V]} deg(x_k) + 1$, $G$ can be colored with at most $\max_{k \in [V]} deg(x_k)$ colors, making it $\max_{k \in [V]} deg(x_k)$-colorable~\cite{lovasz1975three}. For any $G_{d, V}$, $\chi({G_{d, V}})\geq\omega(G_{d, V})$, where $\omega(G_{d, V})$ represents the clique number, i.e., the size of the largest complete sub-graph (clique) in $G_{d, V}$~\cite{lovasz1975ratio}.

Upper and lower bounds on $\chi(G_{d, V})$ hold for ${d, V} \in \mathbb{Z}^+$ with $V\geq d+1$. From Proposition~\ref{prop-deg-regular}, $G_{d, V}^n$ is a regular graph. In $G_{d,V}^2$, a clique of size $d$ exists among $d$ adjacent sub-graphs, leading to $\chi({G_{d, V}^2})=d^2$, as at least $d$ colors are needed per sub-graph. Similarly, in the $j$-fold OR product, each sub-graph connects to $d$ others, forming a $d^j$ clique, resulting in $\chi({G_{d,V}^n})= d^n$. 

\subsection{Proof of Corollary~\ref{cor-expansion_upper_bound_and _lower_bound}}
\label{app:proof_expansion_upper_bound_and_lower_bound}  
For the upper bound in (\ref{eq_expansion_upper_bound_and _lower_bound}), note that for $K_i$ where $i=V$, the largest eigenvalue is $V-1$ with multiplicity 1, and $-1$ has multiplicity $V-1$ \cite{qi2020characterization}. Thus, $\Lambda(K_i^n) = \max (\lambda_2({\bf A}_f^n), |\lambda_{V^n}({\bf A}_f^n)|) = 1$, and using (\ref{eq:expansion}), we obtain the RHS of (\ref{eq_expansion_upper_bound_and _lower_bound}). Subsequently, for the lower bound in (\ref{eq_expansion_upper_bound_and _lower_bound}), we adjust (\ref{eq:expansion}) for $C_i^n$, and determine $\expan(C_i^n)$. For computing $\Lambda(C_i^n)$, we note that $C_i^n$ is a $d$-regular graph as specified in Proposition~\ref{prop-deg-cycles}. Therefore, the eigenvalues satisfy $\{\lambda_1 > \lambda_2 > \cdots > \lambda_{V^n}$\}. Noting that the diagonal entries of ${\bf A}_f^n$ are zero allows us to deduce that $trace({\bf A}_f^n) = 0$. Given that for $C_i^n$, $\lambda_1({\bf A}_f^n)$ is positive and has the largest cardinality, $\lambda_{V^n}({\bf A}_f^n)$ must be approximately equal in magnitude (negative value) and greater in value than $\lambda_2({\bf A}_f^n)$ so that the trace becomes zero, as numerically demonstrated for $C_i^n$ in Section~\ref{sec:choromatic-using-eigenvalue-cycles}, which indicates the cardinality of $\lambda_k$ follows the order $\{d \geq |\lambda_1| > |\lambda_{V^n}|>|\lambda_2|>\cdots\}$. Therefore, by replacing $\Lambda(C_i^n)$ where $\max (\lambda_2({\bf A}_f^n), |\lambda_{V^n}({\bf A}_f^n)|)=|\lambda_{V^n}({\bf A}_f^n)|$, in the denominator of (\ref{eq:expansion}), we reach (\ref{eq_expansion_upper_bound_and _lower_bound}).

\subsection{Proof of Proposition~\ref{prop_entropy_odd_cycles_MIS}}
\label{App:Proof-of-Proposition_prop_entropy_odd_cycles_MIS}
We define the cardinality of MISs for the $n$-fold OR product graph $G^n$, $|\text{MIS}_{G^n}|$, as follows:
\begin{align}
\label{eq-proof-prop-general-entropy}
  |\text{MIS}_{G^n}|=(|\text{MIS}_{G}|)^n \ .
\end{align}

(\ref{eq-proof-prop-general-entropy}) follows from, given $G$ with MISs of size $|\text{MIS}_{G}|$, the $2$-fold OR product contains ISs of size $|\text{MIS}_{G}|$ in each sub-graph, giving $|\text{MIS}_G|^2$ in total, and similarly for higher-order products. Using (\ref{eq-proof-prop-general-entropy}) and (\ref{eq-graph_entropy-korner}), we derive an upper bound on $H_G(X_1)$ given $V^n$ vertices. As in Proposition~\ref{prop_entropy_odd_cycles_MIS}, we apply constraints on $\alpha_t$ for coloring MISs and compute $H(\mathcal{C}_{G^n})$, which is bounded by evaluating $|\text{MIS}_{G^j}|$ for each $j$-fold OR product, $j \in [n]$, under valid coloring constraints. 
 \begin{align}
 \label{eq-proof-prop-ent-constraints}
(a)\: \left(\sum_{t=1}^{n} \alpha_t\cdot\frac{|\text{MIS}_{G}|^t}{V^n}\right)+\alpha_0\cdot\frac{|\text{MIS}_{G^0}|}{V^n}=1, \: \: 
 (b)\: \: \alpha_t\geq |\text{MIS}_{G}|\cdot \alpha_{t-1} \ ,  \: (c)\: \:  \alpha_{0}=1, 
\end{align}
$(a)$ illustrates the coloring PMFs assigned to MISs, where $|\text{MIS}_{G^n}|=|\text{MIS}_{G}|^n$ denotes the size of MISs within sub-graphs of $G^n$. In each sub-graph, $|\text{MIS}_{G^j}|$ represents the maximum independent set size of $G^j(l)$ for $l\in [V]$. Constraints $(b)$ and $(c)$ are identical to those of Proposition~\ref{prop_entropy_odd_cycles_MIS}. By assigning colors to MISs and covering all nodes using $\alpha_t$, we derive a coloring PMF for MISs with cardinality $|\text{MIS}_{G^n}|$, yielding an achievable chromatic entropy bound for any $G$. As $G^j$ changes with the power $j$, the MIS size also changes. For the upper bound, we identify a set of $\alpha_t$ satisfying the minimum entropy coloring \cite{korner1979encode} under the constraint $\alpha_t > \alpha_{t-1}$,
\begin{align}   
\label{eq_MIS_based_entropy_odd_cycle}
H\Big(\alpha_{n}\cdot \big(\frac{|\text{MIS}_{G^n}|}{V^n}\big), \alpha_{n-1}\cdot \big(\frac{|\text{MIS}_{G^{(n-1)}}|}{V^{n}}\big), \cdots, \alpha_0 \cdot\frac{1}{V^n}\Big)  \ .
\end{align}
 
Subsequently, we use (\ref{eq-chorm_korner-cycle}) and normalize the entropy calculated in (\ref{eq_MIS_based_entropy_odd_cycle}) for the $n$-fold OR product by a factor of $\frac{1}{n}$, and use (\ref{eq-proof-prop-general-entropy}), which proves the statement of the proposition.



\subsection{Proof of Corollary~\ref{cor-conv-graph-fraction}}
\label{App:proof-theo-conv-graph-fraction}
For any connected $G(\mathcal{V}, \mathcal{E})$ with $|\mathcal{V}| = 2k+1$, we have $\chi(C_{2k+1}) \leq \chi(G)$, so $\chi(C_{2k+1}^n)$ gives a lower bound on $\chi(G^n)$.
Therefore, we first calculate $\chi_f({C_{2k+1}^n})$, 
then use it to derive a lower bound on $H_{C_{2k+1}}({X}_1)$ and subsequently on $H_{G}({X}_1)$. Using $a:b$ coloring for $C_{2k+1}$, we show that at least 
$2b+1$ colors are needed and claim that the infimum of $\chi_f({C_{2k+1}^n})$ occurs at $b=k$. We prove $\chi_b({C_{2k+1}})=2b+1$, by contradiction: Assume that $\chi_b(C_{2k+1})\leq 2b$ and assign each vertex $i\in [2k+1]$ a coloring set $\mathcal{C}_i$ of size $b$. For $i\in [2k]$, the coloring sets are defined as $\mathcal{C}_1 = [b]$, $\mathcal{C}_2=[b+1, 2b]$, $\dots$, and $\mathcal{C}_{2k} =[b]$. Using $2b$ colors recursively over the vertices, disjoint color sets are assigned to each node up to the last node $x_{2k+1}$. However, reusing the previous $2b$ colors for $\mathcal{C}_{2k+1}$ is impossible, leading to a contradiction of the initial assumption. Next from (\ref{eq-proof-theorem-converse-fractional-number}), we show that $\frac{2k+1}{k} \leq \chi_f({C_{2k+1}})$ holds. From \cite{lovasz1975ratio}, the lower bound on $\chi_f({G})$ is given by $\chi_f({G})\geq \frac{V}{|\text{MIS}_G|}$, which by applying $G=C_{2k+1}$, proves the assumption that $b=k$. 

To compute $H^{\chi}_{{C^n_{2k+1}}}({\bf X}_1)$, we use the bound $H_G^f(X_1) \leq H_G(X_1)$ from~\cite[Lemma 1]{malak2022fractional}, and leveraging their coloring PMF in~\cite[Proposition 2]{malak2022fractional}, yielding as $n \to \infty$:
\begin{align}
\label{eq-coloring-frac-derya}
H^f_{G}(X_1)=\lim_{n\to \infty} \frac{1}{n} \inf_b \frac{1}{b} \min_{\mathcal{C}^f_{C_{i}^n}} H(\mathcal{C}^f_{C_{i}^n}) \ .
\end{align}

Given a uniform ${C_{i}^n}$, the entropy is simplified and calculated as $H^{\chi_f}_{{C^n_{2k+1}}}({\bf X}_1)= \log_2 |\chi_f({C_{2k+1}^n})|$. From (\ref{eq-proof-theorem-converse-fractional-number}), the graph entropy of $C_{2k+1}^n$, incorporating $\frac{1}{n}$ normalization for the $n$-th realization, is:
\begin{align}
 \label{proof-theorem-frac-lower}
\frac{1}{n} H^{\chi}_{{C^n_{2k+1}}}({\bf X}_1)=\frac{1}{n}\cdot\log_2 \left(\frac{2b+1}{b} \right)^n=\log_2 \left( \frac{2b+1}{b} \right) \ ,
\end{align}
where the infimum of $b$ is $k$ in (\ref{proof-theorem-frac-lower}), this captures $H_{C_{2k+1}}({X}_1)$ and validating (\ref{eq-theo-conv-graph-fraction}) for any $G$.

\subsection{Proof of Theorem~\ref{theo-eig-any-func}}
\label{app:proof-eig-any-func}

By Lemma~\ref{eigenvalue_dependency}, the eigenvalues of the sum of two symmetric matrices equal the sums of their respective eigenvalues.
To calculate $\lambda_k({\bf A}_f^n)$ for $k\in [V^n]$, we express ${\bf A}_f^n$ as the sum of two symmetric matrices and determine the eigenvalues by summing those of the individual matrices. For $n=1$, the eigenvalues $\lambda_k({\bf A}_f)$ can be calculated using GCT or QR decomposition methods. For $n=2$, the adjacency matrix ${\bf A}_f^2$ is partitioned into ${\bf A}_f^2={\bf A}_{Gr}^2 + {\bf A}_{f^{\mathsf{c}}}^2$, where ${\bf A}^2_{Gr}$ is a diagonal block with ${\bf A}_f$ on the diagonals and ${\bf Z}_V$ elsewhere. The term ${\bf A}_{f^{\mathsf{c}}}^2 = {\bf A}_f^2 - {\bf A}_{Gr}^2$ is constructed with all-zero matrices ${\bf Z}_V$ on the diagonal and ${\bf J}_V$ or ${\bf Z}_V$ on the off-diagonal blocks based on the sub-graphs connections. This decomposition allows:
\begin{align}
\label{eq-Geresh_proof-sec-eig}
   \lambda_k({\bf A}_f^2)= \lambda_k({\bf A}_{Gr}^2) + \lambda_k({\bf A}_{f^{\mathsf{c}}}^2), \quad k\in [V^2].
\end{align}

Assume that the decomposition holds for $n=j$, ${\bf A}_f^j={\bf A}_{Gr}^j+{\bf A}_{f^{\mathsf{c}}}^j$ and $\lambda_k({\bf A}_f^j) = \lambda_k({\bf A}_{Gr}^j) +\lambda_k({\bf A}_{f^{\mathsf{c}}}^j)$. For $n = j+1$, decompose ${\bf A}_f^{j+1}$ as: ${\bf A}_f^{j+1} = {\bf A}_{Gr}^{j+1} + {\bf A}_{f^{\mathsf{c}}}^{j+1}$. Here, ${\bf A}_{Gr}^{j+1}$ is block diagonal with ${\bf A}_f^j$ on the diagonal blocks, and ${\bf A}_{f^{\mathsf{c}}}^{j+1}$ contains off-diagonal terms derived from ${\bf J}_{V^{j+1}}$ and ${\bf Z}_{V^{j+1}}$. Thus, by induction, $\lambda_k({\bf A}_f^{n-1})$ is computed using $\delta_k$ for using GCT according to (\ref{eq-interval-gresh}), and the decomposition and eigenvalue sum hold for all $n \geq 2$ as shown in (\ref{eig-any-func-eq}).


\subsection{Proof of Corollary~\ref{cor:iter-gresh-eigen}}
\label{App:cor-iter-gresh-eig-general-nth} 
Consider a block diagonal matrix (e.g., ${\bf A}^j_{Gr}\in \mathbb{F}_2^{V^j\times V^j}$, $j\in[n]$), where each diagonal block is identical to ${\bf A}^{j-1}_{f}\in \mathbb{F}_2^{V^{j-1}\times V^{j-1}}$, and all off-diagonal entries are zero. The eigenvalues of ${\bf A}^j_{Gr}$, $j\in[n]$ is equal to $\lambda_k({\bf A}^{j-1}_{f})$, $k\in[V^{j-1}]$, but the algebraic multiplicity of each eigenvalue is scaled by the number of times the block ${\bf A}^{j-1}_{f}$ appears on the diagonal.  
Thus, for the $2$-fold OR product, we have $\lambda_k({\bf A}_f^2) = \lambda_k({\bf A}^2_{Gr}) + \lambda_k({\bf A}_{f^{\mathsf{c}}}^2)$.
Furthermore, for $n=3$, we can substitute $\lambda_k({\bf A}^3_{Gr})$ with $\lambda_k({\bf A}^2_{Gr})+\lambda_k ({\bf A}_{f^{\mathsf{c}}}^2)$, as follows $\lambda_k({\bf A}_f^3)=\lambda_k({\bf A}^2_{Gr})+\lambda_k ({\bf A}_{f^{\mathsf{c}}}^2)+\lambda_k ({\bf A}_{f^{\mathsf{c}}}^3)$. 
Similarly, for $j$-fold OR product, $\lambda_k({\bf A}_f^j)=\lambda_k({\bf A}^2_{Gr})+\lambda_k ({\bf A}_{f^{\mathsf{c}}}^2)+\cdots + \lambda_k ({\bf A}_{f^{\mathsf{c}}}^{j-1})+\lambda_k ({\bf A}_{f^{\mathsf{c}}}^j)$, and for $n=j+1$, the relation is $\lambda_k({\bf A}_f^{j+1})= \lambda_k({\bf A}^2_{Gr})+\cdots + \lambda_k ({\bf A}_{f^{\mathsf{c}}}^{j+1})$. 
Similarly, for the $n$-fold OR product $\lambda_k({\bf A}_f^n)$ is calculated as follows: 
\begin{align}
\label{eq-cor-iter-gresh-eig-n} 
\lambda_k({\bf A}_f^n)=\lambda_k({\bf A}^2_{Gr})+\lambda_k ({\bf A}_{f^{\mathsf{c}}}^2)+\cdots + \lambda_k ({\bf A}_{f^{\mathsf{c}}}^{n-1})+\lambda_k ({\bf A}_{f^{\mathsf{c}}}^n), \: k\in[V^n]\ .
\end{align}

By substituting $\lambda_k({\bf A}^2_{Gr})$ in (\ref{eq-cor-iter-gresh-eig-n}) with $\lambda_k({\bf A}_f)$, we reach the statement of the corollary in (\ref{cor-iter-gresh-eig-general-nth}).

\begin{spacing}{1}
\bibliographystyle{IEEEtran}
\bibliography{ref}
\end{spacing}

\end{document}